\documentclass[10pt,leqno]{article}
\usepackage{graphicx,psfrag,amssymb,amsmath,amsthm,bbm, MnSymbol}
\usepackage[dvips]{color}
\usepackage{psfrag}
\usepackage{eucal}

\newtheorem{theorem}{Theorem}

\newtheorem{proposition}{Proposition}

\numberwithin{equation}{section}

\begin{document}

\title{\bf From billiards to thermodynamics}
\author{
T. Chumley\footnote{Department of Mathematics, Washington University, Campus Box 1146, St. Louis, MO 63130},
\ \  S. Cook\footnote{Department of Mathematics and Statistics, Swarthmore College, College Ave, Swarthmore PA 19081}, \ \
R. Feres\footnotemark[1]
}
\date{\today}

\maketitle

\begin{abstract}
We explore some beginning  steps in  stochastic thermodynamics
of   billiard-like mechanical systems  by introducing extremely  simple  and explicit 
random mechanical processes  capable of exhibiting steady-state  irreversible thermodynamical behavior. In particular,
we describe a Markov chain model   of a minimalistic heat engine and numerically study its 
operation and efficiency. 
 \end{abstract}

\section{Introduction}

The main purpose of this paper is to introduce a class of simple random mechanical model systems
that may  help in  shedding  light on the mechanisms  whereby   steady state, out of equilibrium thermodynamical behavior
can emerge  in random dynamics.  In the spirit of classical statistical mechanics, our  random systems 
  arise as a certain  form of course-graining of Hamiltonian mechanical models; among these models
 we look to the simplest ones,  in which all the interactions are through elastic collisions, namely 
 {\em billiard systems. }
  
We shall  interpret the term ``emerging thermodynamical behavior'' rather concretely 
by  considering the    problem of obtaining  a  random mechanical system
that
can perform at steady-state (stationary) mode as a {\em heat engine},
 defined in an explicit fashion using few degrees of freedom.

The  main motivation lies in the belief that 
   a rich collection  of    model  systems that are amenable 
to  detailed   numerical and  analytical exploration
is 
essential   to guiding  the  development of   a  stochastic theory of non-equilibrium thermodynamics.  See
\cite{qian} for the mathematical outlines of such a theory. For  a more applied perspective
see, e.g.,  \cite{seifert1,seifert2} and related literature on {\em stochastic thermodynamics}.
Our main contribution here is to describe   purely mechanical, billiard-like 
stochastic (Markov chain) processes, obtained
by a specific form of coarse-graining 
from deterministic billiards, 
built from a rather small  number of parts and   fully explicit  in a sense to  be
clarified later in the paper, which can exhibit (in the mean) textbook
thermodynamical behavior.  
In particular, we describe a minimalistic {\em billiard-Markov} heat engine  capable
of producing mechanical  work in a steady-state regime of operation.

Related studies in the stochastic thermodynamics literature,
 or in numerical studies of molecular motors and models of {\em Feynman's ratchet and pawl system}
(such as in \cite{zheng}), 
 typically start from a Langevin equation and the {\em a priori} existence of a  heat bath at a given temperature.
A  distinguishing feature of the present work is that  we have   an explicit Markov model of the heat bath-thermostat,
which is 
 very closely related to the deterministic system  from which it is  derived. 
Our systems are examples of  {\em random billiards},
a term that will be    expanded upon later  in the  paper. (See also \cite{scott,fz,fz2}.)

The paper is organized as follows.  In Section \ref{deterministicbilliards}, the most basic  facts 
about (deterministic) billiard dynamical systems are recalled for later use.
These are mechanical systems in which the interaction between moving parts is limited to  elastic collisions; in particular, there are no potentials or dissipative forces. 
 The section  lingers  a bit on a 
description of the natural flow invariant  measure in the phase space of the billiard system,
emphasizing   the so-called {\em cosine law} for billiard reflections.
 As a prelude to introducing our
billiard thermostat later in Section \ref{billtherm}, we also show in Section \ref{deterministicbilliards}
how the Maxwell-Boltzmann distribution of velocities is obtained from the cosine law and
a simple (well-known) geometric argument concerning  the phenomenon of concentration of volumes   in high dimensions. 

Section \ref{billiardmodels}  introduces the idea of {\em random billiard systems.}
These are Markov chain systems with general (i.e., not necessarily countable) state spaces  obtained from deterministic billiards
(see Section \ref{deterministicbilliards}) by the following general method.  We select one or more dynamical variables
of a  given deterministic mechanical system and turn  them into random variables with
 fixed in time probability distributions. It is natural to choose for the latter   the asymptotic probability  distribution that  those
variables attain  in the original  deterministic system.   The resulting random dynamical system is
often not far removed, in certain ways, from the deterministic system that gave rise to it. For example,
for the main class of random billiards described  in Section \ref{billiardmodels}, the 
velocity factor of the flow-invariant measure in phase space becomes a stationary
measure for the associated random process,  suggesting that  the random and deterministic
systems have closely related ergodic theories. 

Also in  Section \ref{billiardmodels} we introduce and explore a random billiard system that 
will serve as our all-purpose heat bath-thermostat. It is indicated there (and proved in \cite{scott})
that a sequence of  collisions of a point mass with the random billiard thermostat yields  a Markov chain process (in
the state space of post-collision velocities) whose stationary distribution  is
Maxwell-Boltzmann's. 
Thus, in our systems, thermostatic
action is not imposed by fiat but modeled explicitly.  As already noted, this is a distinguishing feature of
 our models. 

All pieces will  then be  in place 
to  study heat flow between two billiard thermostats at different temperatures. 
This is done in Section \ref{heatflowsection}.  The basic mechanism of operation 
of a heat engine already becomes apparent at this point. In Subsection \ref{heatengine}
we propose a particular design for such an engine which is as simple as we could conceive.
The associated random billiard is $5$-dimensional, that is, it contains $5$ moving parts (point masses): 
one ``gas molecule,'' one ``Brownian particle,'' an additional particle that acts as an  ``escape valve''
to the gas molecule, and two masses for the two thermostats at different temperatures.  The whole contraption 
is essentially one-dimensional in physical space. 
We briefly explore the engine's operation by numerical simulation and compute the mean velocities of 
rotation and  the engine's  
(rather modest, but positive) efficiency for different values of  a force load.

 \section{Deterministic billiards}\label{deterministicbilliards}
A brief overview is given here   of the most basic properties of deterministic
 billiards  needed for our discussion of random billiards in the next section. Most
 importantly, is a description of the billiard flow-invariant volume in phase space and
 the so-called {\em cosine law} for billiard reflection.
 \subsection{Basic facts}
Billiard systems, broadly conceived,  are Hamiltonian systems on  manifolds with boundary,
  the boundary points representing {\em collision configurations}. Most commonly, the configuration manifold
  is a    region in the Euclidean plane having piecewise smooth boundary, although higher dimensional 
  systems are widely studied and will be encountered throughout  this paper. 
  Higher dimensional billiards typically describe mechanical systems consisting of several
   rigid constituent masses interacting only through collisions. 
  The configuration manifold  is endowed with
the Riemannian metric defined by the kinetic energy bilinear form. In particular, the (linear) collision
map at boundary points   of the configuration manifold is a linear isometry under the assumption of
energy conservation.  
The collision  map is often taken to be the standard Euclidean {\em reflection}, that is,  a map that  fixes all the vectors tangent to the
boundary while sending a   vector perpendicular to the boundary to its negative. In this paper
 the Riemannian metric on configuration space will always have constant coefficients (associated to masses
of the constituent rigid parts of the system)
and so will be an Euclidean metric.

\vspace{0.1in}
\begin{figure}[htbp]
\begin{center}
\includegraphics[width=2.0in]{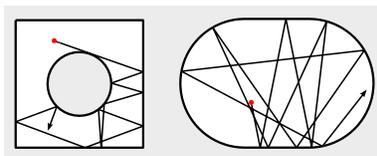}\ \ 
\caption{\small A version of Sinai's Billiard on the left,  and the Bunimovich stadium on the right. These are
two examples of {\em ergodic} billiard systems.}
\label{shapes}
\end{center}
\end{figure}

Figure 
  \ref{shapes} shows two famous  examples of the basic kind of  billiard system. In each case, the {\em billiard table} is
  a planar region   whose boundary consists of   piecewise smooth curves; the {\em billiard particle}   undergoes uniform rectilinear motion in the interior of the region, bouncing
  off specularly after hitting  the boundary.

 In general, let $M$ denote the billiard's configuration manifold. This is 
 the planar regions in the $2$-dimensional examples of Figure \ref{shapes}.  The {\em phase space} is the bundle of tangent vectors $TM$
  on which  one defines the {\em flow map} $\varphi_t$. The flow map assigns  to each time $t$ and tangent vector $(q,v)\in TM$
the  state (i.e., the position and velocity) $\varphi_t(q,v)$ of the billiard trajectory at time $t$  having initial   conditions $(q,v)$ at time $0$.  

 It  will be assumed here that the billiard particle is not subject to
a potential function or any   form of interaction other than  elastic collision.  For a more general
perspective see  \cite{scott}. Thus the speed of billiard trajectories (given in terms
of the mechanically determined Riemannian metric) is a constant of motion,
usually arbitrarily set to $1$,   and 
 the flow map $\varphi_t$ is often restricted to the submanifold of unit vectors in $TM$.
The precise definition of the billiard flow contains 
some important fine print,   dealing with the issue of
singular trajectories;  for example, those trajectories that end at corners or graze the boundary of $M$.
For the omitted details (in dimension $2$)
 see  \cite{chernov}.

A fact of special significance is that the billiard flow map leaves invariant a canonical volume form on phase space. There is also an associated invariant volume form on 
the space of unit vectors on the boundary of $M$. The existence of these invariant volumes is fundamental
 for the ergodic theory of billiard systems and for the probability theory we wish to
employ later, so we take a moment to describe them in detail.

\vspace{.1in}
\begin{figure}[htbp]
\begin{center}
\includegraphics[width=2.0in]{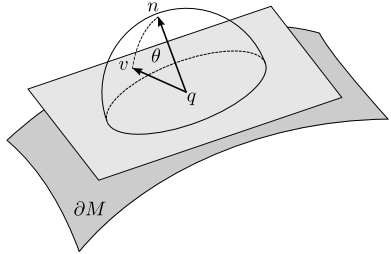}\ \ 
\caption{\small A piece of the boundary of a billiard region, showing the unit hemisphere at
a point $q$. The unit normal vector $n$ points to the interior of the $d$-dimensional manifold  $M$ and $v$ is a unit tangent vector
to $M$ at $q$ forming an angle $\theta$ with $n$.  If $d\omega$ denotes the $(d-1)$-dimensional volume on
the unit hemisphere  at the boundary point $q$ of  $M$, 
then $d\nu=\cos\theta\, d\omega$ is the  factor of the invariant volume 
accounting for velocities at $q$.}
\label{tangent}
\end{center}
\end{figure}

Let $d$ be the dimension of $M$
and $S^+$  the subset of $TM$ consisting of
unit vectors at boundary points of $M$ pointing towards the interior of $M$.  Then $S^+$ is the disjoint union
of hemispheres $S^+_q$ defined at each $q\in M$. The unit normal vector $n_q$ is contained in $S_q^+$;
we denote by $\theta$ the angle between a given $v\in S^+_q$ and $n_q$ and by $d\omega(v)$ the
$(d-1)$-dimensional volume element at $v$ over $S^+_q$. Also let $dV(q)$ denote the
volume element at $q$ on the boundary of $M$ (associated to the induced Riemannian metric).
The billiard flow $\varphi_t$ induces a map $T$ on $S^+$ as follows:  for each $v\in S_q^+$ write 
$(q(t), v(t))=\varphi_t(q,v)$, where $t$ is  the moment of next collision with the boundary; then 
$T(q, v):=(q(t), \overline{v}(t))$, where $\overline{v}$ indicates  the reflection of $v$ back into $S^+$. 
We refer to $T$ as the {\em billiard map}. 
The transformation $T$ is said to {\em preserve}, or {\em leave invariant} a measure $\nu$ on $S^+$ if (writing
$u=(q,v)$)
  $$\int_{S^+} f(u)\, d\nu(u)=\int_{S^+} f(T(u))\, d\nu(u)$$
for every  integrable function $f$.  The next proposition is well-known.
\begin{proposition}
$T$ leaves invariant  on $S^+$ the measure  element 
\begin{equation} d\nu(q,v):=\cos\theta\, dV(q)\, d\omega(v).\end{equation}
\end{proposition}
For a proof (of a   more general expression) under much more general conditions   that allow
for potentials and non-flat  Riemannian metrics see \cite{scott}.

The existence of this invariant measure on $S^+$  is the starting point of the ergodic theory of billiard systems. 
We always assume that the measure is finite and typically rescale  it so that     $\nu(S^+)=1$; in 
this case it is natural to interpret $\nu$ as a probability measure.
A billiard system is said to be {\em ergodic} if $S^+$ cannot be decomposed as a disjoint union
of two measurable subsets, both  invariant under $T$ and  having positive measure relative to 
 $\nu$.  Ergodicity can also be expressed in terms of the equality of time and space means:
\begin{equation}\label{ergodicequality}\lim_{N\rightarrow \infty} \frac1N \sum_{i=0}^{N-1} f(T^i(q,v))=\int_{S^+}f(q,v)\, d\nu(q,v),\end{equation}
where $f$ is any integrable function on $S^+$. (See \cite{katok} for a general reference for  ergodic theory as
a chapter in the mathematical  theory of dynamical systems.)
The existence of the limit, and the equality in \ref{ergodicequality}  under the ergodicity  assumption,  is the content of the celebrated ergodic theorem of Birkhoff.  Below, we refer to the identity itself as
the {\em ergodic theorem}.

\vspace{0.1in}
 \begin{figure}[htbp]
\begin{center}
\includegraphics[width=2.5in]{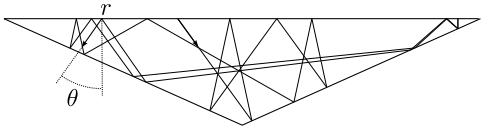}\ \ 
\caption{\small  With probability $1$, the set of return points to a piece of the boundary of an
ergodic billiard satisfies the cosine law: the post-collision angles $\theta\in [-\pi/2,\pi/2]$ have
the distribution $d\mu(\theta):=\frac12 \cos\theta \, d\theta$. The set of positions, indicated by $r$ in the figure,
are distributed uniformly. Polygonal (and polyhedral) billiard tables, as in the figure,  will often appear below, although
it is not well understood when such billiards are ergodic. See \cite{gutkin} for further remarks.
}
\label{triangletrajectory}
\end{center}
\end{figure}

Proving that a   billiard system is  ergodic is generally a technically difficult task.
In fact, a significant part   of the general theory of dynamical systems, particularly hyperbolic (strongly chaotic) systems, has
been developed in pursuit of establishing ergodicity for such  statistical mechanical systems as   hard spheres   models of a gas. (See, e.g., \cite{bali}  or \cite{buni}, chapter 8.)

An immediate consequence of the ergodic theorem is that the long term distribution of post collision angles
of  an ergodic billiard in any dimension satisfies the  cosine law, whereas the 
distribution of collision points on the boundary of $M$ is uniform relative to the measure $dV$. 
More precisely,  let $v_1, v_2, \dots$ be the velocities immediately after collisions registered at
each moment that  a billiard trajectory returns to a  segment of the boundary of $M$ having positive measure.
 Then for almost all initial conditions the set of angles is distributed
according to 
$d\mu(v)= C \langle n,v\rangle \, d\omega(v)$, where $C$ is a normalizing constant and the angle brackets 
denote inner product and  $\langle n,v\rangle=\cos\theta$, where $\theta$ is the angle between
$v$ and $n$.

\subsection{Illustrating the cosine law with a variant of Sinai's billiard}\label{divided}
The billiard table of Figure \ref{billiardtrajectory1} represents a container divided in  two chambers by a porous solid screen
 composed of small circular scatterers. The scatterers are separated by small gaps.
  A billiard particle   represents a spherical gas molecule.
One is interested, for example, in   how a ``gas'' consisting of a large
 number of   billiard particles injected at time $t=0$ into, say, the left chamber,
 will expand to fill up the entire container.

This  billiard  table can be regarded  as an ``unfolding''
 of {\em Sinai's billiard} shown on the left of Figure \ref{shapes}, and from  this observation
it can be shown that the associated billiard flow  is ergodic.  
 Figure \ref{billiardtrajectory1}
 shows one long segment of trajectory, indicating the initial velocity vector and the image of that
 vector under the billiard flow at time $t$. 
This is an example of a {\em (semi-) dispersing billiard}, which  are 
well-studied models of chaotic dynamics (see \cite{chernov}). Trajectories are highly unstable
 in their dependence on initial conditions due to the presence of the circular scatterers.
 
 \vspace{0.1in}
 \begin{figure}[htbp]
\begin{center}
\includegraphics[width=2.0in]{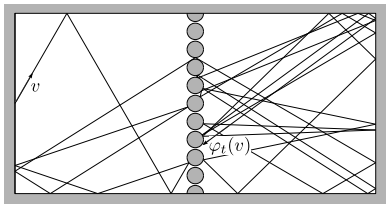}\ \ 
\caption{\small A billiard model of a container divided by a solid porous screen consisting
of small circular scatterers separated by small gaps.}
\label{billiardtrajectory1}
\end{center}
\end{figure}

 Consider 
  Figure
 \ref{cell},
 where we focus on one fundamental cell of the solid screen.
We  define the
{\em reduced phase space} of this system as the set $$S=\{0,1\}\times [0,1]\times [-\pi/2,\pi/2].$$
A state of the form $(k, r, \theta)$ gives the initial condition of a trajectory that
enters into the scattering region from the left ($k=0$) or the right ($k=1$) chamber 
at a position $r$  in the interval $[0,1]$, with velocity
$v=(-1)^k \cos\theta e_1 + \sin\theta e_2$, where $e_1$ and $e_2$ are the standard basis vectors of $\mathbb{R}^2$.
The reduced billiard map $T:S\rightarrow S$ then  gives  the end state of a trajectory that begins
  and is stopped at $S$. The billiard motion on the full table is an appropriate
composition of  $T$ 
with a similar return map on a   rectangular table.

   \begin{figure}[htbp]
\begin{center}
\includegraphics[width=1.7in]{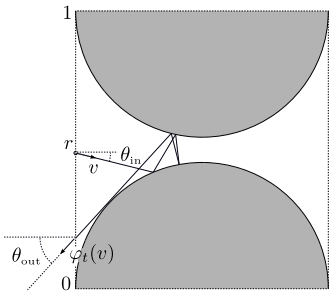}\ \ 
\caption{\small The core of the dynamics of the divided chambers billiard can is in the motion near a fundamental cell of the scattering screen. 
}
\label{cell}
\end{center}
\end{figure}

Given a long trajectory of a billiard particle, we register the values $k_1, k_2, \dots$ in $\{0,1\}$,
which is  the sequence of sides of the container the particle occupies 
at each moment it enters the scattering region;   $r_1, r_2, \dots$ in $[0,1]$,  the sequence of 
positions along the flat   boundary segments  of the fundamental cell at which  the particle enters the region;
and $\theta_1, \theta_2, \dots$ in $[-\pi/2, \pi/2]$,    the sequence of angles the particle's velocity makes
with the normal vector to those boundary segments. A remark about the first sequence will be observed
shortly;  first note that the long term distribution of the $r_i$ is uniform along the unit interval. This
follows from the above observation on  the form of the invariant measure and the ergodic theorem, and
is  observed in the  numerical experiment of Figure \ref{entry}.

\begin{figure}[htbp]
\begin{center}
\includegraphics[width=3.5in]{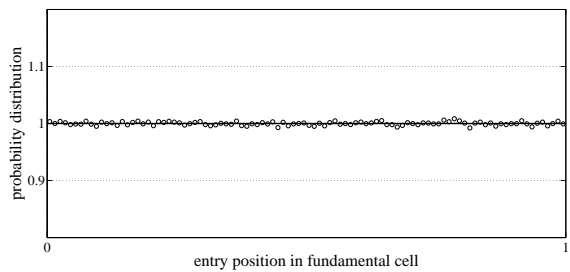}\ \ 
\caption{\small Long term distribution of entry positions   into a fundamental cell of the scattering screen. The graph was obtained by numerically simulating the billiard motion
over a period of $10^7$ entries into the scattering region. 
}
\label{entry}
\end{center}
\end{figure}

The distribution of the angles $\theta_i$ is given, as expected, by the cosine law. This is shown in Figure \ref{cosine}.
For ergodic polygonal billiards  these long term distributions of positions and angles hold, but 
convergence is much slower.

Figure \ref{experiment} shows the result of releasing a large number of  independent (i.e., that do not collide with each other) billiard particles
at  $t=0$ into the left chamber of the container. The solid line  graph gives
the fraction of particles in the right chamber as a function of time. (Time is expressed in arbitrary units of length divided by
$100$. Recall that the speed is set equal to $1$.) The other  graphs are explained later. (Section \ref{billiardmodels}.)

\begin{figure}[htbp]
\begin{center}
\includegraphics[width=3.5in]{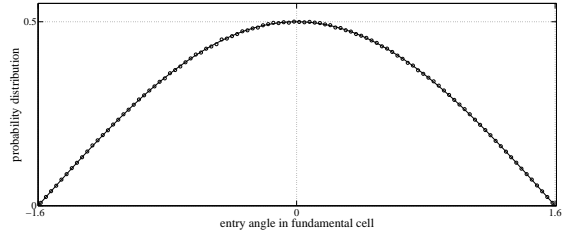}\ \ 
\caption{\small Long term distribution of the entry angles   into a fundamental cell.
(Details  as for Figure \ref{entry}.)
}
\label{cosine}
\end{center}
\end{figure}

The salient point this graph   dramatizes is the issue of time reversibility versus irreversibility. 
In the long run the fraction of particles in each chamber appears to stabilize  to $50\%$, as
our physical intuition would suggest. This behavior of the system of many particles introduces  a sense of  direction
of  passage of  time  that is not present in the time reversible nature of the billiard dynamic.
The issue of explaining irreversible behavior in the collective motion of a large number
of particles whose fundamental evolution is time reversible  is a central problem in 
non-equilibrium  statistical physics. This so-called {\em arrow of time} problem, of
deriving macroscopic irreversibility from microscopic reversibility, has
bedeviled the study of statistical mechanics  since its beginnings  in Boltzmann's fundamental work in the early
1870s.
In fact,  one early objection to Boltzmann's work is 
the so-called {\em Zermelo's paradox} \cite{steckline}, which is based on  a fundamental  observation of H. Poincar\'e known as {\em Poincar\'e's recurrence} (see
\cite{katok} for its abstract, measure theoretic form) implying that, with probability $1$ on the initial conditions,  there will be an infinite sequence of times  
when those $10^7$ particles all come together back into the initial chamber on the left.
 We refer the reader to the vast literature on
the Boltzmann equation and the $H$-theorem for more information on this topic.

\begin{figure}[htbp]
\begin{center}
\includegraphics[width=3.5in]{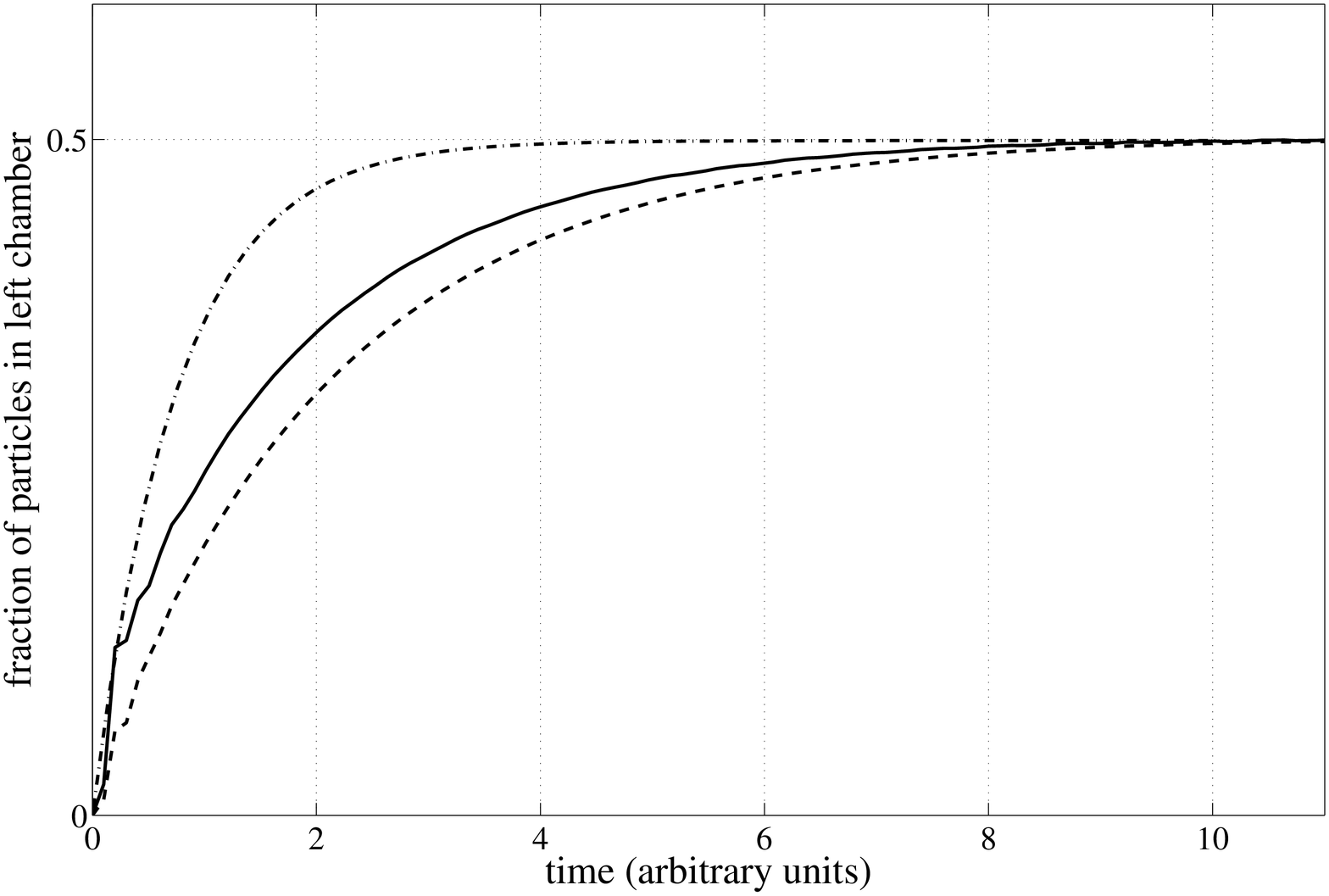}\ \ 
\caption{\small
$10^7$ (non-interacting) billiard particles are released from the middle of the left wall on  the left-hand   chamber 
uniformly over the range of angles $-\pi/4$ to $\pi/4$ relative to the positive $x$ axis. The container
is $20$ units  long by $9$  tall, and the spacing between scatterers is $1$. 
Time unit is $100$. The solid line refers  to the deterministic billiard of Figure \ref{billiardtrajectory1},
 the  dashed line (below the first) is the same   for the associated random billiard (introduced in the
 next section), and the
 dash-dot line is the corresponding plot for the two-state Markov chain for which the transition probabilities
 between chambers and the mean time of transition were obtained numerically.}
\label{experiment}
\end{center}
\end{figure}

The cosine distribution of angles has great significance in kinetic theory of gases and gas diffusion, particularly
in the so-called {\em Knudsen regime} of rarified gases, when transport properties are more strongly  affected by
collisions between gas molecules and the walls of the container than by collisions between the molecules themselves.
The appearance of the cosine term in the scattering distribution of gas molecules was studied experimentally by
M. Knudsen in the early years of kinetic theory. His experiments are described in \cite{knudsen}.
In many texts, particularly in engineering, the cosine distribution is often referred to as   {\em Knudsen's cosine law}.
See \cite{fy} for further information.

\subsection{A geometric remark about many particles  systems}\label{rescale}
The single particle billiard system is a geometric representation of a mechanical  system that 
may consist of many constituent rigid particles interacting with each other through elastic collisions. 
This  simple remark is immediately understood by considering  the two-particle, one-dimensional
billiard system shown at the top  of  Figure \ref{twoparticletriangle}.

To be fully specified, the billiard table must be given a Riemannian  metric relative to which reflections
are specular.  The triangular region of Figure \ref{twoparticletriangle} with the standard Euclidean inner product
does  not in general  define a  billiard system   since
if $m_1\neq m_2$, 
 the single particle in the triangle, whose $x$ and $y$ coordinates
give the positions of the two masses along the interval $[0,L]$, will  not reflect specularly when colliding with
the diagonal side of the triangle.  A simple way to make the collision 
specular is to absorb the mass values into the position coordinates.  Thus we define coordinates
$x_{\text{\tiny new}}=\sqrt{\frac{m_1}{m}} x,\  y_{\text{\tiny new}}=\sqrt{\frac{m_2}{m}} y$
where $m=m_1+m_2$, and note that the kinetic energy of the system, expressed in the new coordinates,
is a constant multiple of the ordinary Euclidean norm. 
Therefore, a linear transformation that conserves energy becomes an orthogonal  map.  Conservation of linear momentum
means that the component of the pre-collision velocity vector  in the direction of the slanted side of the
triangle in the new metric equals the same component for the post-collision velocity. 
Therefore, the normal component of the pre- and post-collision velocities can only be either  equal or the negative
of each other. Obviously, the latter must the case as there would be no collision otherwise.

\begin{figure}[htbp]
\begin{center}
\includegraphics[width=2.0in]{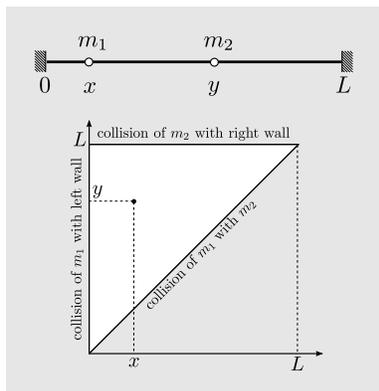}\ \ 
\caption{The billiard table of the two-particle system.}
\label{twoparticletriangle}
\end{center}
\end{figure}

These new, mass-rescaled  coordinates yield a bona fide billiard system on the plane. 
We call the single particle system in the triangular region with the new metric the {\em billiard representation} of the one-dimensional two-particle system. The idea is obviously very general and works in any dimension, for any number of masses.
In higher dimensions, say, for the collision of two solid bodies in $3$-dimensional space, 
the basic conservation laws of energy, linear and angular momentum, as well as the imposition of time-reversibility
and linearity,
do not fully specify the collision map. Further assumptions about the nature of contact, such as being slippery or rubbery,
are needed.

\subsection{Knudsen implies Maxwell-Boltzmann}
One has not  entered thermodynamics until temperature is somehow brought into the   picture, and for our needs
  this may be  done via  the Maxwell-Boltzmann distribution of velocities.  In the present  section we
illustrate with a simple example the  geometric explanation  of  how this  fundamental distribution  arises in the   context of billiard dynamics. This discussion serves  to motivate  our model of random billiard
 thermostat that will be introduced  in Section \ref{billtherm} and is not strictly necessary for  defining the billiard heat engine
 of Section \ref{heatengine}. The reader who wishes to  skip this section on first reading may do so without great loss of continuity.

 \vspace{0.1in}
\begin{figure}[htbp]
\begin{center}
\includegraphics[width=2.5in]{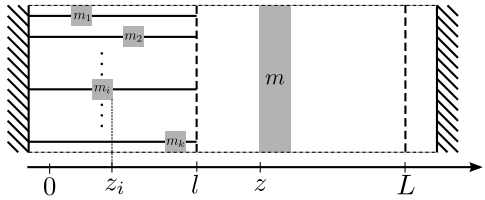}\ \ 
\caption{\small  A billiard model that helps explain the origin of the Maxwell-Boltzmann distribution of
scattered velocities. }
\label{multiple1}
\end{center}
\end{figure} 

The example is  shown in Figure \ref{multiple1}.  It  consists of   point masses $m_1, \dots, m_k, m$ that
can slide without friction  on a  line.
Masses $m_i$  are restricted to lie in the interval $[0,l]$  and
they move independently of each other. Their position coordinates are  
 indicated by $z_i$;
$m$ can move in the bigger interval $[0,L]$, with position coordinate  $z$. 
At the endpoints of  $[0,l]$
the 
  $m_i$ bounce  off   elastically. 
  Mass $m$
moves  freely past $l$ (dashed line in Figure \ref{multiple1}), and it collides elastically with the $m_i$ and with the wall at $z=L$.
We imagine the $m_i$ as  tethered to the left wall by  inelastic and massless,  but fully flexible strings of length $l$; when the strings are stretched to the  limit of their length, the masses
 bounce back
as if hitting  a solid wall  at $l$.

To make the system more symmetric without changing
it in any essential way, we regard  the wall on the left  as  a mirror and we keep track of both $z_i$ and its image $-z_i$; thus
 $z_i\in [-l,l]$ can be  negative.  (The thickness of the masses is considered negligible in this model.) In this symmetric form, the billiard representation of the system is as shown   in  Figure \ref{tent}.

Let $M=m+m_1+\dots +m_k$. 
Changing coordinates to $$x_i=\sqrt{m_i/M}\, z_i \text{ for } i=1,\dots, k,  \text{ and } x_{0}=\sqrt{m/M} z,$$ the kinetic energy form  becomes
$$ K(x,\dot{x})=({M}/2) \left(\dot{x}_0^2+\cdots +\dot{x}_{k}^2\right).$$
We may equivalently assume that $(x_1, \dots, x_k)$ defines a point on the hypercube  with coordinates
 $x_i$ in $I_i:=[-a_i/2,a_i/2]$, where $a_i=2\sqrt{m_i/M}\, l$,
 having in mind the above   comment about mirror image.
  Mass  $m$ is then constrained to move on the interval 
  $F(x_1, \dots, x_k)  \leq x_{0}\leq \sqrt{m/M} L$ where 
$$F(x_1, \dots, x_k):=\max\left\{\sqrt{m/m_1}\, |x_1|, \dots, \sqrt{m/m_{k}}\, |x_k|\right\}.$$

 \vspace{0.1in}
\begin{figure}[htbp]
\begin{center}
\includegraphics[width=1.7in]{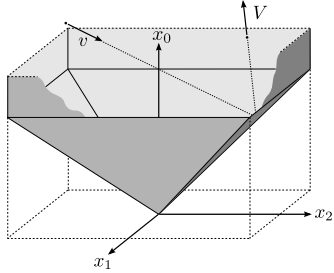}\ \ 
\caption{\small Part of the billiard representation of the system of Figure \ref{multiple1} for $k=2$
describing the interaction between $m$ and the masses $m_1$ and $m_2$.}
\label{tent}
\end{center}
\end{figure} 
Thus the configuration manifold is $$M=\left\{(x_0, x)\in I_1\times \cdots \times I_k\times \mathbb{R}:  F(x)\leq x_{0}\leq \sqrt{m/M}L\right\},$$
and collision is represented (due to energy and momentum conservation and time-reversibility),  by
specular reflection at the boundary of $M$. We now wish to follow the motion of mass $m$; geometrically,
this amounts to following  the image of billiard orbits under the orthogonal projection $\pi$, as indicated in Figure \ref{projection}.
In particular, what can be said about the distribution of values of the projection $\overline{v}(t):=\pi(v(t))$ 
of the velocity of typical billiard trajectories over long time spans? 
The following elementary proposition points to an answer.

\begin{proposition}\label{projectMB}
Let $S:=S_+^k\left(\sigma\sqrt{k+1}\right)$ denote the  hemisphere of dimension $k$ and radius $\sigma\sqrt{k+1}$, consisting of vectors
$v=(v_0, \dots, v_{k})\in \mathbb{R}^{k+1}$ such that $v_{0}>0$ and $v_0^2+\dots +v_{k}^2=(k+1)\sigma^2$. Let
$\mu_k$ be the Knudsen cosine probability measure on $S$; thus $d\mu_k(v)=C_k v_0 \, dV(v)$, 
where $V$ is the Euclidean  volume measure on $S$ and $C_k$ is a normalizing constant. 
Let $\nu_k$ be the image of $\mu_k$ under the projection map $\pi(v)=v_0$. Thus  the $\nu_k$-measure of
an interval $A\subset \mathbb{R}$ is, by definition, $\nu_k(A):=\mu_k\left(\{v: \pi(v)\in A\}\right)$.
Then, as $k$ goes to infinity, the sequence of $\nu_k$ converges (in the vague topology of probability measures)
to $\nu$ on $(0,\infty)$ such that
\begin{equation}\label{MBpc} d\nu(v_0)=\frac{v_0}{\sigma}\exp\left(-\frac12v_0^2/\sigma^2\right) dv_0.\end{equation}
We refer to $\nu$  as the {\em post-collision Maxwell-Boltzmann} probability measure in dimension $1$, with   parameter
$\sigma^2$. Similarly, let $\nu_{i,k}$ be the probability distribution of $v_i$, $i\neq 0$, given that $v$ is
distributed according to $\mu_k$. Then in the limit as $k$ approaches infinity $\nu_{i,k}$ converges
to the Gaussian  
\begin{equation}\label{MBrt} d\nu_{i}(v_i) = \frac{\exp\left(-\frac12 v_i^2/\sigma^2\right)}{\sigma\sqrt{2\pi}}\, dv_i.\end{equation}
\end{proposition}
\begin{proof}
For convenience set  $R:=\sigma\sqrt{k+1}$ and let $S^{k-1}$ be  the unit $(k-1)$-sphere in $\mathbb{R}^k$ centered at the origin. 
Let  $\phi:[0,R]\times S^{k-1}\rightarrow S^k_+(R)$ be the polar coordinates map on the hemisphere, which is defined by
 $$\phi(v_0, v)=\left(v_0, \sqrt{R^2-v_0^2}\, v\right). $$
Let  $dV_{S^k_+(R)}$ denote the volume form on the $k$-dimensional  hemisphere of radius $R$ and $dV_{S^{k-1}}$
the volume form on the unit sphere of dimension $k-1$. A geometric exercise yields the expression of 
$dV_{S^k_+(R)}$ in the just defined  coordinates as
$$ dV_{S^k_+(R)}=R\left(R^2-v_0^2\right)^{\frac{k-2}2}\, dv_0\,  dV_{S^{k-1}}.$$
Given now any bounded  function $f(v_0)$ on the interval $[0, R]$, we obtain by a change of variables in
integration that
\begin{align*}\int_0^R f(v_0)\, d\nu_k(v_0)&=\int_{S^k_+(R)} f(\pi(v))\, d\mu_n(v)\\
&=C_n\int_0^R\int_{S^{k-1}} f(v_0)v_0 R\left(R^2-v_0^2\right)^{\frac{k-2}{2}}\, dv_0 \, dV_{S^{k-1}}.\end{align*}
Integrating over the unit $(k-1)$-sphere in the last integral gives, for a new constant  $D_k$,
$$\int_0^R f(v_0)\, d\nu_k(v_0))=D_k\int_0^R f(v_0)v_0 R\left(R^2-v_0^2\right)^{\frac{k-2}{2}}\, dv_0.$$
Reverting back to $R=\sigma\sqrt{k+1}$ and using that $(1+a/m)^m$ converges to $e^a$ as $m$ tends to infinity, 
finally gives (for yet another constant $C$ independent of $f$)
$$\lim_{n\rightarrow \infty} \int_0^{\sigma\sqrt{k+1}} f(v_0)\, d\nu_n(v_0)=C\int_0^\infty f(v_0) v_0 \exp\left({-\frac12 v_0^2/\sigma^2}\right)  dv_0.$$ As $f$ is arbitrary  we conclude that
$$d\nu(v_0)=Cv_0 \exp\left({-\frac12 v_0^2/\sigma^2}\right)  dv_0. $$
The constant $C$ is easily found to be $1/\sigma$ by normalization. 
The claim for the other components of $v$ is similarly demonstrated.
\end{proof}

Proposition \ref{projectMB} is a manifestation of the well-known connection between probability theory (and statistical physics)
and geometry in high dimensions. An especially intriguing exposition of this connection under the
heading of {\em concentration of measures} may be found in \cite{gromov}, chapter $3\frac12$.

\vspace{0.1in}
\begin{figure}[htbp]
\begin{center}
\includegraphics[width=2.0in]{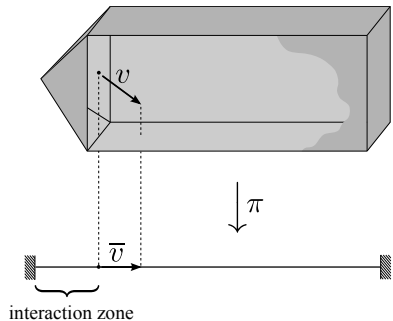}\ \ 
\caption{\small  Orthogonal projection from the multidimensional billiard to the one-dimensional reduction that tracks
the motion of the single gas molecule. Between the leftmost  dashed-line segment and right-hand end of the
interval, the projected   motion is uniform; collision with the right-hand wall is
ordinary one-dimensional billiard reflection. }
\label{projection}
\end{center}
\end{figure}

The appearance of the Maxwell-Boltzmann (MB) distribution in our billiard model can now be explained as follows. 
Observing  the velocity of the mass $m$ amounts to taking the projection $\pi$ of the velocity of the billiard trajectory
as in Figure \ref{projection}; if the billiard system is ergodic (this depends on the ratios of masses, although
as far as we know there is no general criterion  of  ergodicity for polyhedral billiard tables, even in dimension $2$),
then as indicated earlier the long term distribution of velocities $v_1, v_2, \dots$ at the moments $t_1, t_2, \dots$
when the billiard particle emerges from the interaction zone on the left-hand side of the polyhedral table follows
a cosine distribution. The proposition now implies that the projections $\pi(v_1), \pi(v_2), \dots$ should then
follow the approximate MB distribution for finite $n$. The approximation becomes better
as the number of masses near the wall of the system of Figure \ref{multiple1} increases and the total energy
increases proportionally.

\vspace{.1in}
\begin{figure}[htbp]
\begin{center}
\includegraphics[width=3.0in]{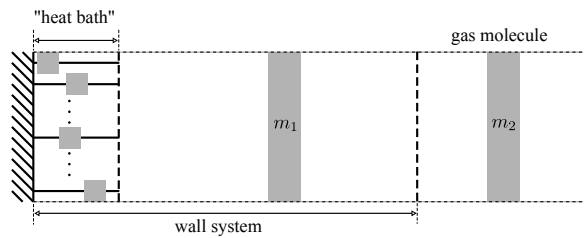}\ \ 
\caption{\small A billiard model of wall with  thermostatic properties. }
\label{multiple_second}
\end{center}
\end{figure}

Reverting to the initial velocity variables (i.e., before we absorbed the masses to form the above $v_i$)  
and indicating by $v$ the velocity of $m$, the post-collision MB distribution
can be written as
\begin{equation}\label{MB} \rho_{\text{\tiny MB}}(v)= \beta m v \exp\left(-\beta \frac{m v^2}2\right)\end{equation}
where $\beta$ is a parameter with units of energy. Later on, after we introduce our
random billiard model for a thermostat, we will remark on how equality of $\beta$ for
two parts of a system is a necessary condition for stationarity, so we recover the 
idea of thermal equilibrium. In statistical physics ones writes  $\beta = 1/kT$, where $k$ is the
so-called  {\em  Boltzmann constant} and $T$ is absolute temperature.

Notice the difference between what we have called above the ``post-collision'' MB distribution 
and the MB distribution for the particle's velocity sampled at random times, in which case the velocity can be both
positive and negative. If $\rho_{\text{\tiny MB}}(v_0)$ is the post-collision density shown in \ref{MB},
then at a random time each velocity $v_0$ should be weighted by the time  the particle, having  this velocity,
takes to go from one end of the interval to the other, which is proportional to $1/v_0$. This term cancels out the
factor $v_0$ in $\rho_{\text{\tiny MB}}(v_0)$, yielding  the standard one-dimensional MB-distribution \ref{MBrt}.

We point out for later use the   model  of Figure \ref{multiple_second}
of thermal interaction between gas molecule and wall. 
The  $k$ 
masses on the far left have a very short range  of motion,
limited by 
the first dashed line,
 compared to $m_1$, which is limited by the second dashed line on the right.
 The gas molecule, $m_2$, can move across those lines.
  As discussed above,
when  the  number  of masses constituting the finite ``heat bath''  grows, 
the asymptotic distribution of positions of $m_1$ (under the assumption of ergodicity)   becomes uniform  and
the distribution of velocities of $m_1$ becomes Gaussian. 
The  random billiard thermostat to be  introduced in the next section will 
be abstracted from this deterministic model by eliminating the masses on the left (the ``heat bath'')
and setting the statistical state of  $m_1$ equal to  the asymptotic distribution (of position and velocity)
this mass would have  in  the deterministic system in the limit of very large $k$.

\section{Random billiard models and the billiard thermostat}\label{billiardmodels}
Our thermodynamical systems will be defined as stochastic processes derived from 
billiard systems. The central concept is of a {\em random billiard}, explained below. 
After general definitions and motivations, we introduce a model thermostat, which is
the key component in the construction of our heat engine in the next section. 
\subsection{Random billiards}
A very current and active program in the ergodic theory of hyperbolic (chaotic) systems, in particular
chaotic billiard systems, is dedicated to obtaining probabilistic limit theorems such as the central limit
theorem for the deterministic system. 
This is a  technical area of investigation, which we do not attempt to survey here. Our 
goal is to abstract from the billiard systems plausible random models that
we can more easily study and of which more explicit results can be derived. Thus we now
turn to the topic of {\em random billiards}.

The basic idea is as follows.  Starting with a deterministic billiard system, we take some of its dynamical
variables and assume that they are random variables with a given distribution. As will be seen in the specific
examples, the resulting system is typically expressed as a Markov chain with non-discrete state space. 
The selection of variables and the choice of probability law assumed for them can vary, but
we observe the following procedure: {\em the probability law for a given random variable is taken to
be the asymptotic distribution that that variable assumes in the  deterministic system from which the random
system is derived.}

To illustrate this idea we return to the divided chamber example of Section \ref{divided}.
There are many possibilities for turning the original system into a random system; we first indicate
an extremely coarse model and then show a much more refined one. 
The coarse model, which only serves a didactical purpose and  is not going to be of further use, 
 consists of a two-state Markov chain with state space $\{0,1\}$, where
$0$ stands for the left side chamber and $1$ for the right side one. We set a time unit $\tau$
equal to the mean time a billiard trajectory takes to return to the zone of scatterers and calculate (numerically)
the transition probabilities of moving from one to the other chamber at each return. 
The resulting Markov chain is shown in diagram form in Figure \ref{markovchambers}.

\vspace{.1in}
\begin{figure}[htbp]
\begin{center}
\includegraphics[width=3.0in]{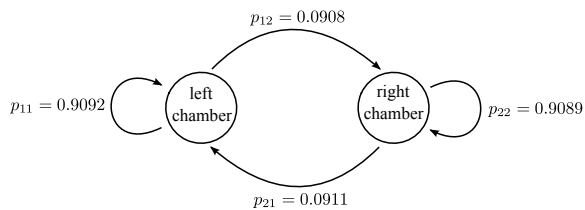}\ \ 
\caption{\small
Numerically obtained transition probabilities between chambers.
Here we use the indices  $1$ for the left chamber and $2$ for the right chamber. }
\label{markovchambers}
\end{center}
\end{figure} 

The dot-dashed   graph of Figure \ref{experiment} shows the result of the experiment of releasing
a large number of particles in one chamber and observing how long it takes for
the distribution of  particles  to even out. The solid line
gives the same distribution for the original deterministic system.

We now turn to the  more refined model (Figure \ref{randombilliard1}) which, as will be seen,  preserves many of the geometric features of the original system.
 The screen of circular scatterers is replaced with a vertical  line.
 Upon colliding with this line,
the billiard particle changes both direction and chamber as prescribed by   transition probabilities  
 with state space $S=\{0,1\}\times [-\pi/2,\pi/2],$ where
the first factor indicates as before the side of the divided container ($0$ for left and  $1$ for  right)
and the second factor gives the angle along which the particle  impinges on or scatters off    the dividing screen. 
Recall the deterministic map $T$ defined on the reduced phase space
 $\{0,1\}\times[0,1]\times [-\pi/2,\pi/2]$ of the fundamental cell shown in Figure
 \ref{cell}.  
The  velocity and chamber  of the billiard particle immediately after collision with the scattering line 
are then defined to be    random variables  obtained
  from  $T$  and the   pre-collision side and angle variables by letting 
the position   $r\in [0,1]$ be 
  random,  uniformly distributed over the unit interval.

To obtain the  transition probabilities operator, we
refer back to   the notation set  in Figure \ref{cell}. We wish to describe the
transition probabilities kernel on  $S$ as a family of 
probability measures $\mu_{k, \theta}$ indexed by the elements of $S$. If $f$ is
any bounded measurable function on  $S$, then by definition the conditional expectation
of $f$ evaluated on the post-collision state, given the pre-collision state $(k_-,\theta_-)$ is 
$$(Pf)(k_-,\theta_-):=\int_{S}f(k_+,\theta_+)\, d\mu_{k_-,\theta_-}(k_+,\theta_+):=\int_0^1f(\varphi_T(k_-, r, \theta_-))\, dr$$
where $T=T(k_-, r,\theta_-)$ is the return time to the entry of the fundamental cell (the dashed lines of
Figure \ref{cell}), $r\in [0,1]$ is the position coordinate along either of the entry line segments, and
$\varphi_t$ is the billiard flow given in terms of the unit velocity angle rather  than the velocity vector.

\vspace{0.1in}
\begin{figure}[htbp]
\begin{center}
\includegraphics[width=2.2in]{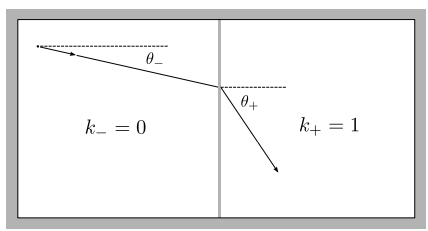}\ \ 
\caption{{\small A  random billiard model for the divided container  experiment. 
The screen of circular scatterers is replaced with  a scattering line. }}
\label{randombilliard1}
\end{center}
\end{figure}

Thus, in this model of random billiard we have replaced the screen of scatterers by a line segment 
separating the two chambers and 
a scattering (Markov) operator  $P$ that updates the direction of the velocity at every
collision with that line segment.  It turns out that the operator $P$ has many nice properties. 
First, the measure $\overline{\mu}$ which assigns probability $1/2$ to $k=0,1$ and the
cosine distribution to $\theta$ turns out to be the unique stationary distribution for $P$.
Second, $P$ can be defined on the Hilbert space of square-integrable functions on $S$
with the measure $\mu$, where it is a self-adjoint operator of norm $1$. 
We refer to   \cite{scott,fz,fz2} for more information about similar operators and their spectral theory.

\begin{figure}[htbp]
\begin{center}
\includegraphics[width=2.5in]{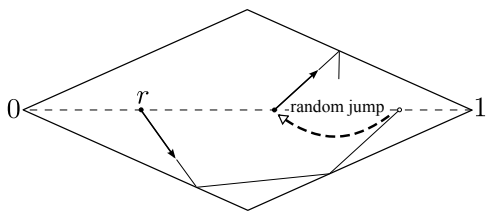}\ \ 
\caption{\small  In this model of a random billiard, the particle trajectory follows ordinary billiard motion until
it crosses the dashed line, at which moment it jumps to a random point along that line, keeping its 
velocity unchanged. The distribution of random position along the dashed line at each crossing is uniform. }
\label{trianglerandom}
\end{center}
\end{figure}

A possible point of concern is that we have illustrated the use of the Markov operator, and claimed that
the cosine law is stationary for it, for a dispersing (Sinai-type) billiard, whereas most of the billiard models
in this paper   are going  to be polygonal or polyhedral, for which ergodicity is
hard to ascertain.  With this in mind, we conclude this section with a much simpler but similar example
of a random billiard  on a 
parallelogram. 
The details are   in Figure \ref{trianglerandom}.

The velocity of the billiard particle as it emerges from the lower triangle of Figure \ref{trianglerandom} through the dashed line is a function of
the velocity as it comes into the triangle and the position $r$. By making $r$ a random variable, the outgoing velocity
becomes a random function of the velocity coming in. We can again describe this velocity response by an operator
$P$ very similar to the one of the previous example (except that the variable $k$ is not present here).
As before, the cosine distribution of exit angles is stationary for the process $V_1, V_2, \dots$, where
$V_i$ is the unit velocity at the $i$th exit from the lower triangle. See  Figure
\ref{cosinetriangle}.

 \begin{figure}[htbp]
\begin{center}
\includegraphics[width=3in]{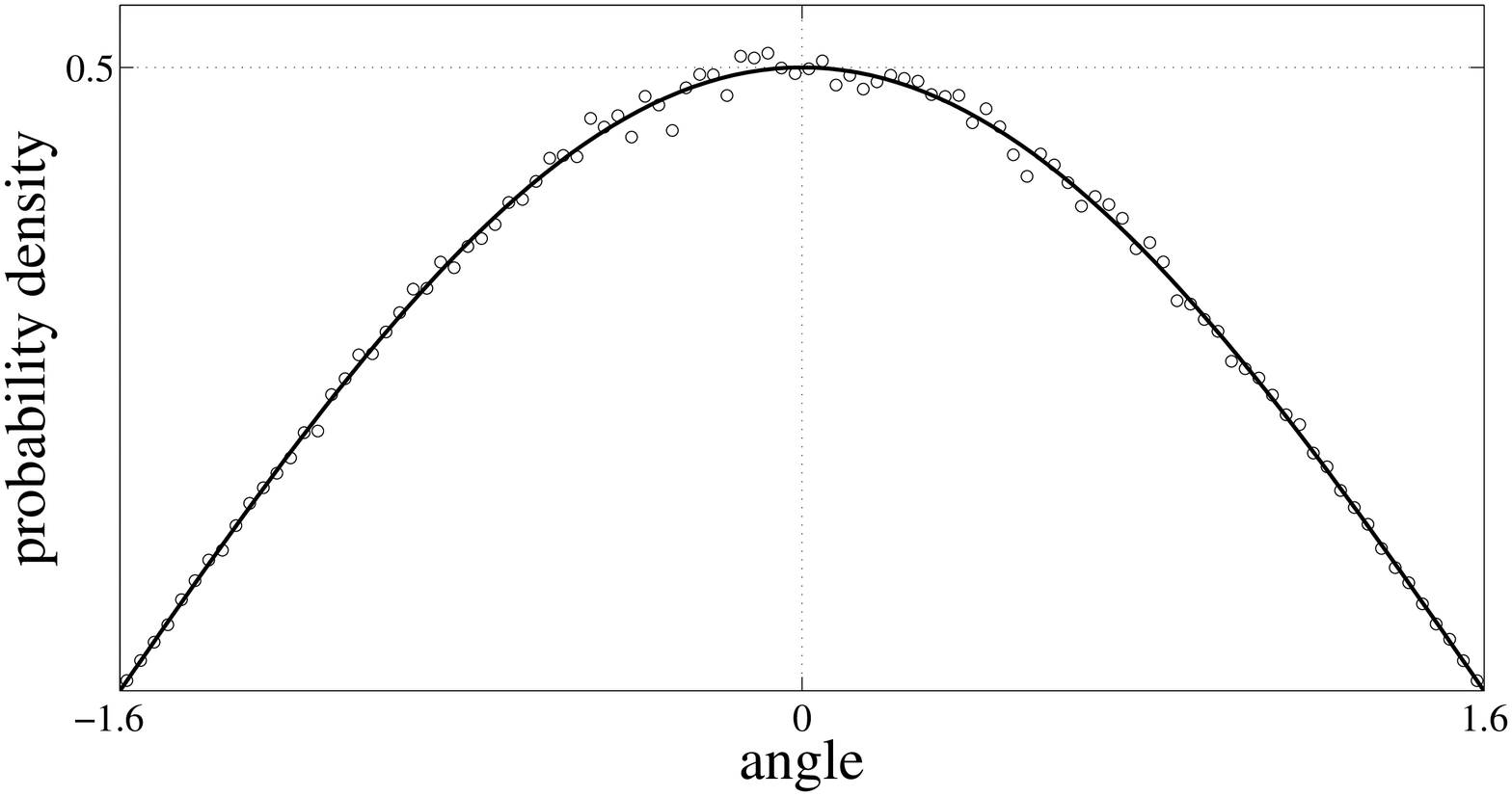}\ \ 
\caption{\small  Histogram of exit velocity angles from the lower  triangle in the random billiard system of Figure
\ref{trianglerandom}, computed for a random orbit with $10^8$ crossings of the dashed line. Convergence is much
slower than for random billiards derived from dispersing deterministic billiards such as the divided chamber example.}
\label{cosinetriangle}
\end{center}
\end{figure} 

  \subsection{A billiard thermostat}\label{billtherm}
We now introduce a random billiard model that will serve as our all-purpose thermostat at a fixed temperature.  
 The details are explained in   Figure
\ref{onemolthermostat} and 
the random  billiard representation of   the system is described in Figure \ref{simplegas}.

This   one-dimensional random billiard thermostat  is  a random reduction of the deterministic system 
of Figure \ref{multiple_second}; the main idea is to eliminate the many masses that are tethered to the left wall,
keeping $m_1$,  and assuming  that  the position and velocity of $m_1$  just prior to interacting with the  gas molecule
are distributed according to its  asymptotic distribution of position and velocity as part of
the   deterministic system of Figure \ref{multiple1}.

 In \cite{scott} we have  studied  Markov chains associated to this system in great detail, including 
some aspects of the spectral properties of the associated Markov operator $P$.  
The state space is now the half-line $(0,\infty)$ of possible values of the velocity of $m_2$ as
it emerges from the interaction zone $[0,l]$ after each iteration of the collision process. 
To write $P$ explicitly  first define  $\gamma=\sqrt{m_2/m_1}$
and write  $P=P_\gamma$   to keep in mind the dependence of the process on this key parameter. 
Recall  that $P_\gamma$ depends on the choice of 
a Gaussian distribution with mean zero and a  $\sigma^2$, representing
the distribution of velocities of $m_1$. (See Proposition \ref{projectMB}.)

 \vspace{0.1in}
\begin{figure}[htbp]
\begin{center}
\includegraphics[width=4.0in]{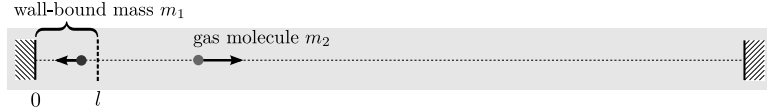}\ \ 
\caption{{\small A  random billiard reduction of  the system of Figure \ref{multiple1} or \ref{multiple_second}.
Mass $m_2$ moves freely over the interval until colliding  with   the wall-bound masses $m_1$.
 We imagine the  latter as    tethered to the wall  by a string of length $l$,
just as  the masses $m_i$ in Figure \ref{multiple1}.
The position of  $m_1$ is assumed to be random uniformly   over   $[0,l]$,
and the velocity is  random normally distributed. At the moment
 $m_2$ crosses the dashed line and thus enters the zone where it can collide with $m_1$,
we choose the state of the latter  (its position and velocity) from its   fixed probability
distributions. From that point on we follow the deterministic motion of the two masses
until $m_2$ leaves $[0,l]$. Prior to every  future collision  
 the statistical state of $m_1$ is reset.}}
\label{onemolthermostat}
\end{center}
\end{figure}

Then  $P_\gamma$ acts on, say, bounded continuous functions on the interval $(0,\infty)$ or, dually,
on probability measures on that interval according to 
 $$(P_\gamma f)(v)=\frac1{l_1}\int_0^{l_1}\int_{-\infty}^{\infty} f(\varphi_\tau(r, (w,v))) \frac{\exp\left(-\frac12 w^2/\sigma^2\right)}{\sigma\sqrt{2\pi}} \, dw\, dr$$
where  the following notation is being used: $\tau=\tau(r, w, v)$ is the return time to the entry (dashed-line) side of
the triangle of Figure \ref{simplegas}  given that the (deterministic) billiard trajectory begins
at   $r\in [0,l_1]$, where $l_1=l\left(m_2/m\right)^{1/2}$, $m=m_1+m_2$, 
$\varphi_\tau$ is the billiard flow stopped at time $\tau$, and $(w,v)$ is the initial velocity of the billiard particle in
dimension $2$. Notice that we are here using the mass-rescaled coordinates  as explained in 
Section \ref{rescale}.

\vspace{.1in}
\begin{figure}[htbp]
\begin{center}
\includegraphics[width=3.0in]{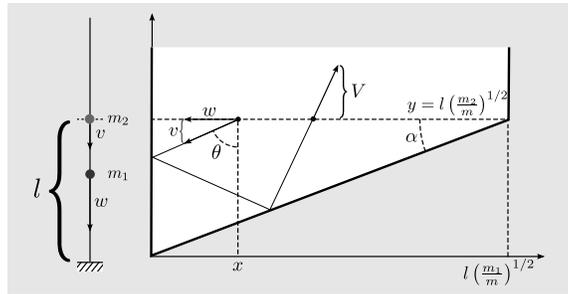}\ \ 
\caption{\small Billiard representation of the system of Figure \ref{onemolthermostat}.
 }
\label{simplegas}
\end{center}
\end{figure}

This amounts to giving  the post-collision velocity
 $V$ of $m_2$ by the following procedure. (See Figure \ref{simplegas}.)
When $m_2$ crosses the line into the zone of free motion of $m_1$,   the horizontal component $w$ of the billiard particle 
is chosen according to 
a Gaussian distribution with mean $0$ and a given variance  $\sigma^2$, and the position along the
upper side of the triangle indicated by a dashed line is chosen to be random uniform. The trajectory  afterwards 
is ordinary, deterministic billiard motion. The outgoing velocity
of $m_2$ is then the vertical component of 
the velocity of the billiard particle as it emerges 
out of the triangle.

  \vspace{0.1in}
 \begin{figure}[htbp]
\begin{center}
\includegraphics[width=3.0in]{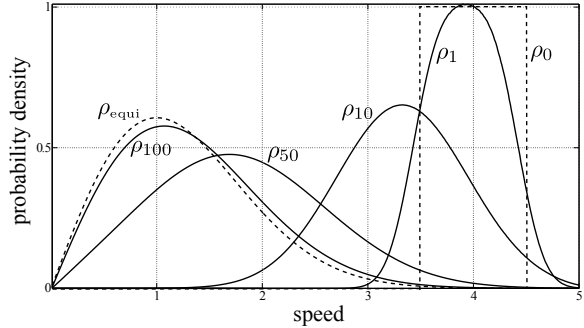}\ \ 
\caption{\small Evolution to  equilibrium of an initial probability measure  for the velocity of the free mass.
(Figure \ref{onemolthermostat}.)
The  $\rho_i$ are the probability densities of the velocities
 $V_i$  of $m_2$  immediately
after each collision with the wall system.
The limit density is the post-collision MB-distribution.  We have used  
a finite rank approximation of $P_\gamma$ obtained by numerically simulating the  system
with parameter $\gamma=0.1$.
 }
\label{MBevolution}
\end{center}
\end{figure}

The basic properties of the billiard-thermostat just defined are listed in the next theorem,
 taken from \cite{scott}. See the cited paper for a  proof. 
The theorem characterizes  the stationary distribution of velocities of the billiard-thermostat
Markov chain and  gives some indication of  how an arbitrary initial distribution  convergences to the unique stationary one. (An estimate
of 
the rate of convergence 
 in terms of the mass ratio 
can also  be found in \cite{scott}.)
  
\begin{theorem}[\cite{scott}]\label{thermtheor}
The following assertions hold for   $\gamma<1/\sqrt{3}$:
\begin{enumerate}
\item   $P_\gamma$ has a unique stationary distribution $\mu$. Its  probability density is
$$ \rho(v)=\sigma^{-1}v\exp\left(-\frac{v^2}{2\sigma^2}\right).$$
\item $P_\gamma$ is a self-adjoint operator of norm $1$ defined on the Hilbert space
of square integrable functions on $(0,\infty)$ with the measure $\mu$. It is  a Hilbert-Schmidt operator
and has, therefore, a discrete spectrum of eigenvalues.
\item  For an arbitrary initial probability distribution $\mu_0$, we have 
$$\|\mu_0 P_\gamma^n-\mu\|_{TV}\rightarrow 0 $$ exponentially fast in the total variation norm.
\end{enumerate}
\end{theorem}

Figure \ref{MBevolution} illustrates the convergence of a sequence of velocity distributions 
under the billiard-thermostat Markov chain process.

   Reverting   to
the original variables (prior to mass-scaling),
stationarity  for the  process of a sequence of  successive collisions between $m_2$ and
the wall system (containing  the mass $m_1$)  implies: (1)  the particle follows a MB distribution and (2) 
$m_1\sigma_1^2=m_2\sigma_2^2$ holds, where $\sigma_1^2$ and $\sigma_2^2$ are the variances, respectively,
of the velocity distribution of $m_1$ (fixed)  and of $m_2$; the latter  evolves from some
initial statistical state toward this equilibrium.  In what follows, any reference to a value of {\em temperature} 
of a wall should be understood as a fixed value $T=m_1\sigma_1^2$.


The equilibrium state 
described in Theorem \ref{thermtheor}
is arrived at by iterating a random map on $(0,\infty)$ with transition probabilities operator
$P_\gamma$. We  show  this map   explicitly here since it is used for
the actual simulation of the process.  
Let $\gamma:=\sqrt{m_2/m_1}=\tan \alpha$, where $\alpha$ is the 
angle  indicated on Figure \ref{simplegas}.  Define 
$$a:=\frac{1-\gamma^2}{1+\gamma^2}, \ \ b:=\frac{2\gamma}{1+\gamma^2},\  \ \overline{a}:=\frac{1-6\gamma^2 +\gamma^4}{(1+\gamma^2)^2}, \ \
\overline{b}:=\frac{4\gamma(1-\gamma^2)}{(1+\gamma^2)^2}.$$ 
Define the functions
$$p_v(w):= \frac{\gamma}{\sqrt{1+\gamma^2}}\frac{|w|}{v}, \ \ q_v(w):= \frac{2(1-\gamma^2)}{1+\gamma^2} -\frac{4\gamma}{1+\gamma^2}\frac{|w|}{v},$$
and introduce the partition of $(0,\infty)$ into intervals $I^i_w=|w| I_i$, $i=1,2,3,4$, where
$$I_i:=(\tan((i-1)\alpha),\tan(i\alpha)], i=1,2,3,\ \  I_4:=(\tan(3\alpha),\infty).$$
To simplify the description of the map, let $m_1>3 m_2$ (equivalent to  $\alpha<\pi/6$).
 Choose $w\in \mathbbm{R}$ at random with probability $\zeta$ and define the
affine maps 
$$F^w_1(v):=av+bw,\ \ F^w_2(v):=av - bw, \ \ F^w_3(v):=-\overline{a}v + \overline{b}w. $$
Finally, let $F^w:(0, \infty)\rightarrow (0,\infty)$ be the  piecewise affine random map defined 
on each interval $I^i_w$ of the partition as follows. 
Case I: If $w\geq 0$, then $F^w(v)=F^w_1(v)$.  Case II: If $w<0$, then
\begin{align*}  
\left.F^w\right|_{I^1_w}(v)&:= F^{|w|}_1(v)   \\
 \left.F^w\right|_{I^2_w}(v)&:=\begin{cases} F^{|w|}_1(v)  &  \text{with  probability } p_v(w)\\
F^{|w|}_3(v)  & \text{with  probability  } 1-p_v(w) \end{cases} \\
\left.F^w\right|_{I^3_w}(v)&:= 
\begin{cases}F^{|w|}_1(v) & \text{with  probability  } p_v(w)  \\  F^{|w|}_2(v) & \text{with  probability  } q_v(w)
 \\
F^{|w|}_3(v) & \text{with  probability  } 1- q_v(w)-p_v(w) 
\end{cases}  \\
\left.F^w\right|_{I^4_w}(v)&:=
\begin{cases}F^{|w|}_1(v) & \text{with  probability  } p_v(w) \\  F^{|w|}_2(v) & \text{with  probability } 1-p_v(w)
 \end{cases} 
\end{align*}
These   are  obtained  by a  tedious but straightforward work.   A
 ``collision between point mass $m_2$ and a wall with temperature $T$'' will
later be interpreted mathematically as  an iteration of $F$ where
the variance of  $\zeta$ is  $T/m_1$.

\section{Heat flow and the billiard-Markov heat engine} \label{heatflowsection}
Now that  temperature has been introduced  into our billiard-Markov models,
 the plan is to explore    basic ideas in  thermodynamics
aimed to build our minimalistic random billiard model of a heat engine. 
References  made  to a ``wall at temperature $T$'' should be understood in
terms of the billiard thermostat model of Section \ref{billtherm} and the random map $F$ given there.

\subsection{Heat flow}

We first discuss heat transport mediated by collisions.

\vspace{0.1in}
\begin{figure}[htbp]
\begin{center}
\includegraphics[width=3.5in]{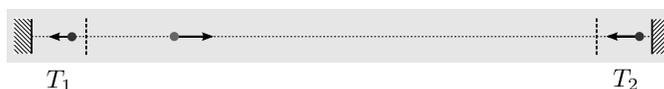}\ \ 
\caption{{\small   Two-sided version of the system  of Figure \ref{onemolthermostat} with
two  different temperatures.
}}
\label{twotempsbill}
\end{center}
\end{figure} 

Consider the experiment described in Figure \ref{twotempsbill}.
 As the wall-bound mass $m_1$  will be  fixed, we may  identify the wall temperature with the variance parameter
  $\sigma^2$ of the velocity distribution of $m_1$.
 The middle particle, of  mass $m_2$,
 will be referred to as the {\em gas molecule}.

 \vspace{0.1in}
\begin{figure}[htbp]
\begin{center}
\includegraphics[width=3.5in]{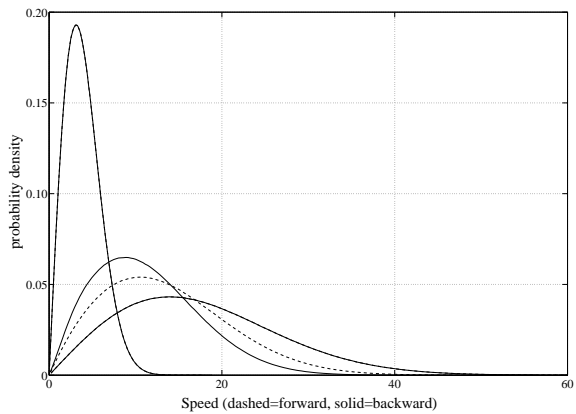}\ \ 
\caption{\small  The figure contains $6$ graphs, only $4$ of which are distinguishable. One pair (the tallest curve) 
gives the probability distributions of the forward and backward velocities 
when the two temperatures are equal and relatively small: $T_1=T_2=\tau_{\text{\tiny cold}}$.
  Similarly, the  shortest pair of graphs
corresponds to equal but relatively high temperature: $T_1=T_2=\tau_{\text{\tiny hot}}$. The two graphs 
in  between show the same distributions when $T_1=\tau_{\text{\tiny hot}}$ and $T_2=\tau_{\text{\tiny cold}}$.
Parameters   used: $m_1=10$, $m_2=1$, the number of iterations (collisions with either side)
was $5\times 10^7$ and the hot and cold temperatures are given by the
variances $\sigma^2_{\text{\tiny hot}}=20$ and   $\sigma^2_{\text{\tiny cold}}=1$.}
\label{experiment1twotemps}
\end{center}
\end{figure}

 We first wish to understand what happens to the stationary velocity  distribution of 
the 
gas molecule. Figure \ref{experiment1twotemps} shows the main effect.
The key observation is that the mean velocity going away from the warmer wall
is greater than the mean velocity moving toward it. This means that
energy is being transferred  from the warmer wall to the colder one through the back and forth motion
of the free mass. 
 The statistical states of the
walls  being  constant, this creates a stationary heat flow between the walls mediated by
the  free  particle.

Let $Q^\text{\tiny hot}_i$ and  $Q^\text{\tiny cold}_i$, for $i=1, 2, \dots$, be
the change in energy of the gas molecule before and after   each collision, alternately  with the hot (say, left) and
cold (right) walls, indexed by the collision number $i$.  Unsurprisingly, it is observed numerically that 
the expected value of  the $Q^\text{\tiny cold}_i$ over a large number of collisions is the negative
of  the expected value of the $Q^\text{\tiny hold}_i$. Furthermore, this expected value, denoted $\overline{Q}^\text{\tiny hot}$,
depends linearly  on the difference of  temperautres:
$$\overline{Q}^\text{\tiny hot}=c(\gamma)\, (T_{\text{\tiny hot}}-T_{\text{\tiny cold}}) $$
where $c(\gamma)$ is a constant which, experimentally,   appears to depend only 
on the main parameter $\gamma$ of the wall-gas molecule system.  
Figure \ref{heatflow} gives some evidence for this  linear  relation.  Each line shows  the
mean energy transferred from the hot wall to the gas molecule for a given
value of $\gamma$. We have set in each case $\sigma^2_{\text{\tiny cold}}=1$
whereas $\sigma^2_{\text{\tiny cold}}$ varied from $1$ to $11$.  
The graphs where virtually the same after shifting both temperatures by an equal  value.

The stage is now  set to try to extract work from this heat flow.  The natural idea
is to take some of the difference in momentum between the forward and backward motion of
the gas molecule and impart it on another mass, which we shall  refer to as the {\em Brownian particle},  to produce 
coherent motion.  
\vspace{0.1in}
\begin{figure}[htbp]
\begin{center}
\includegraphics[width=3in]{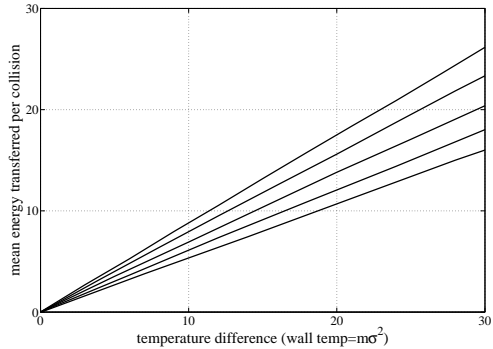}\ \ 
\caption{\small  From top to bottom, $m_2=1,\,  m_1=3.001, \dots, 7.001$.
The mean energy transferred  for each pair of temperatures (expressed by the values of $m_1\sigma^2$) and each value of $\gamma$
was obtained by averaging over  $3\times 10^5$ collisions. 
}
\label{heatflow}
\end{center}
\end{figure} 

\subsection{Description of the  billiard-Markov heat engine}\label{heatengine}
Among the  many possible designs  of  a heat engine built from  
the    billiard-Markov thermostat, we describe here (Figure \ref{motor1})
 the simplest we could devise. 
It consists of two parallel rail tracks, one a short distance above the other. 
The upper track contains a sliding mass $m_2$ (the {\em gas molecule}) and a wall, one side of which is kept at temperature
$T_1$ and the other at temperature $T_2$.  
These  walls can only  exchange energy   indirectly    through 
collisions with the sliding mass.

 \vspace{0.1in}
\begin{figure}[htbp]
\begin{center}
\includegraphics[width=4.0in]{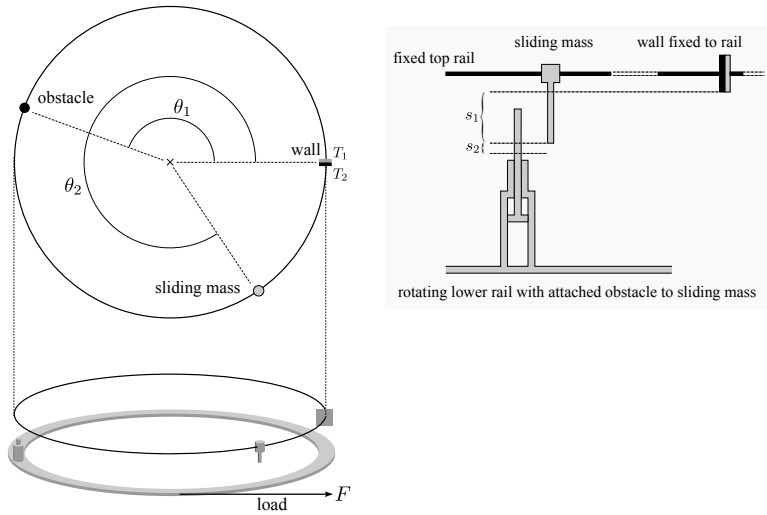}
\caption{\small   A minimalistic  Brownian motor.   }
\label{motor1}
\end{center}
\end{figure} 

The gas molecule moves freely at uniform speeds when
it is away from the wall; when a collision with the wall occurs  we use our model thermostat to obtain the
the post-collision velocity. The lower rail, of mass $m_1$, will be called the {\em Brownian particle.}
 (When running the engine later on, we typically assume   the Brownian mass   to be several times bigger
  than $m_2$.)  It can   rotate freely,  and attached to it is   a protruding  pin  that can move up and down
in billiard fashion; that is, it moves freely within a short vertical interval, bouncing  off elastically against the
limits of the interval.

The maximum height of the pin does not exceed  the lowest point of the wall, so it
never collides with the wall, but it may collide with the gas molecule depending
on how far extended it is.  Therefore, the Brownian particle can at any time be at two possible
states: either ``open'' to the passage of the gas molecule or ``closed'' to it. The times   $\tau_1, \tau_2$
during which  it is closed or open, respectively, alternate periodically as the vertical motion of the pin is assumed
not to be affected by the horizontal motion of the system. These times only depend on the
speed of vertical motion and the lengths $s_1$ and $s_2$ (Figure \ref{motor1}.)

We shall refer to this whole apparatus as the {\em Brownian engine}, or occasionally the {\em billiard-Markov engine}.
The reader will notice some   similarities 
  with the well known {\em Feynman's  ratchet wheel}, although the present design is much
simpler. Alternatively,  the obstacle with a moving
tip can  be regarded as a  piston with an  escape valve  as in an internal combustion engine. 
The contraption may also suggest a distant  relation to  the  Crookes radiometer.
The billiard representation of the Brownian engine is shown in Figure \ref{brownengine1}.

\vspace{0.1in}
\begin{figure}[htbp]
\begin{center}
\includegraphics[width=3in]{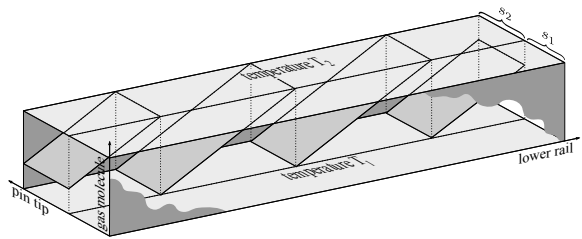}\ \ 
\caption{\small  The $3$-dimensional billiard channel associated to the Brownian engine of Figure \ref{motor1}.
The bottom and top walls reflect at temperatures $T_1$ and $T_2$. This means that the    two components 
of the velocity tangent to those walls are kept the same, while the normal component
is prescribed by the thermostat's random map (see end of Subsection \ref{billtherm}).
All the other walls reflect specularly
(after appropriately rescaling the position coordinates as explained in the text).
We are interested in the projection of the motion along the axis labeled ``lower rail.'' }
\label{brownengine1}
\end{center}
\end{figure}

The whole system contains $5$ moving parts: the gas molecule, the Brownian particle,
the moving tip of the obstacle, and one particle bound to each side of the wall. Thus
$5$
dimensions are required 
for a full description of the random billiard system, but by not showing the billiard structure of the thermostats
we can present it  in dimension $3$.  The variable of special interest is the long axis labeled as ``lower rail''
giving the rotation of the Brownian particle. When later testing the engine  we will want to
add a constant  force $F$ tangential to the  rail so as to investigate the engine's  ability   to
do work (i.e., rotate) agains this force.

Figure \ref{trajectory1} shows a short segment of trajectory. 
 It is apparent  that collisions with the top and bottom sides are not specular and may
not preserve the particle's speed. Collisions with the diagonal sides, when they
occur, are specular.

\vspace{0.1in}
\begin{figure}[htbp]
\begin{center}
\includegraphics[width=4in]{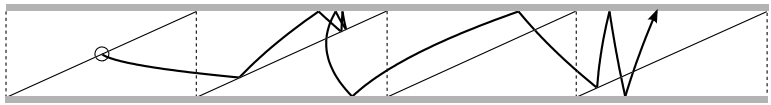}\ \ 
\caption{\small  
Two-dimensional projection of a small segment of trajectory of the Brownian motor with a force load, obtained
by numerical simulation. The circle indicates the
beginning. It is apparent from the curvature
of trajectory segments that the force is acting towards the right-hand side. Distances are rescaled by the masses, so reflections with the diagonal walls are specular (when such reflections occur). }

\label{trajectory1}
\end{center}
\end{figure}

\subsection{The engine's operation; first law  and efficiency }

The typical behavior of the engine, first with $0$ load, is  shown in Figure \ref{brownianpaths}.

\vspace{0.1in}
\begin{figure}[htbp]
\begin{center}
\includegraphics[width=2.9in]{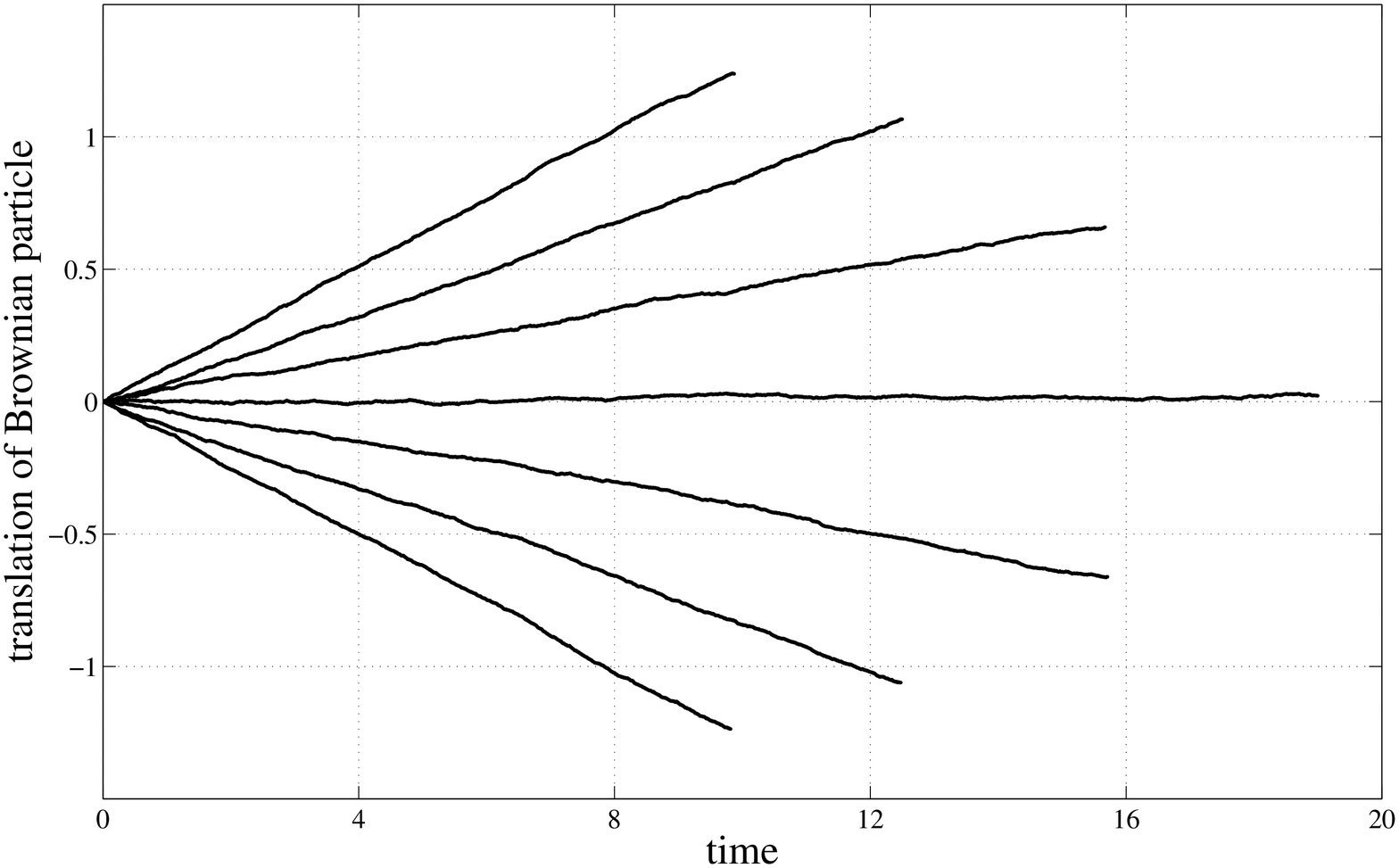} \hspace{-0.2in} \includegraphics[width=1.83in]{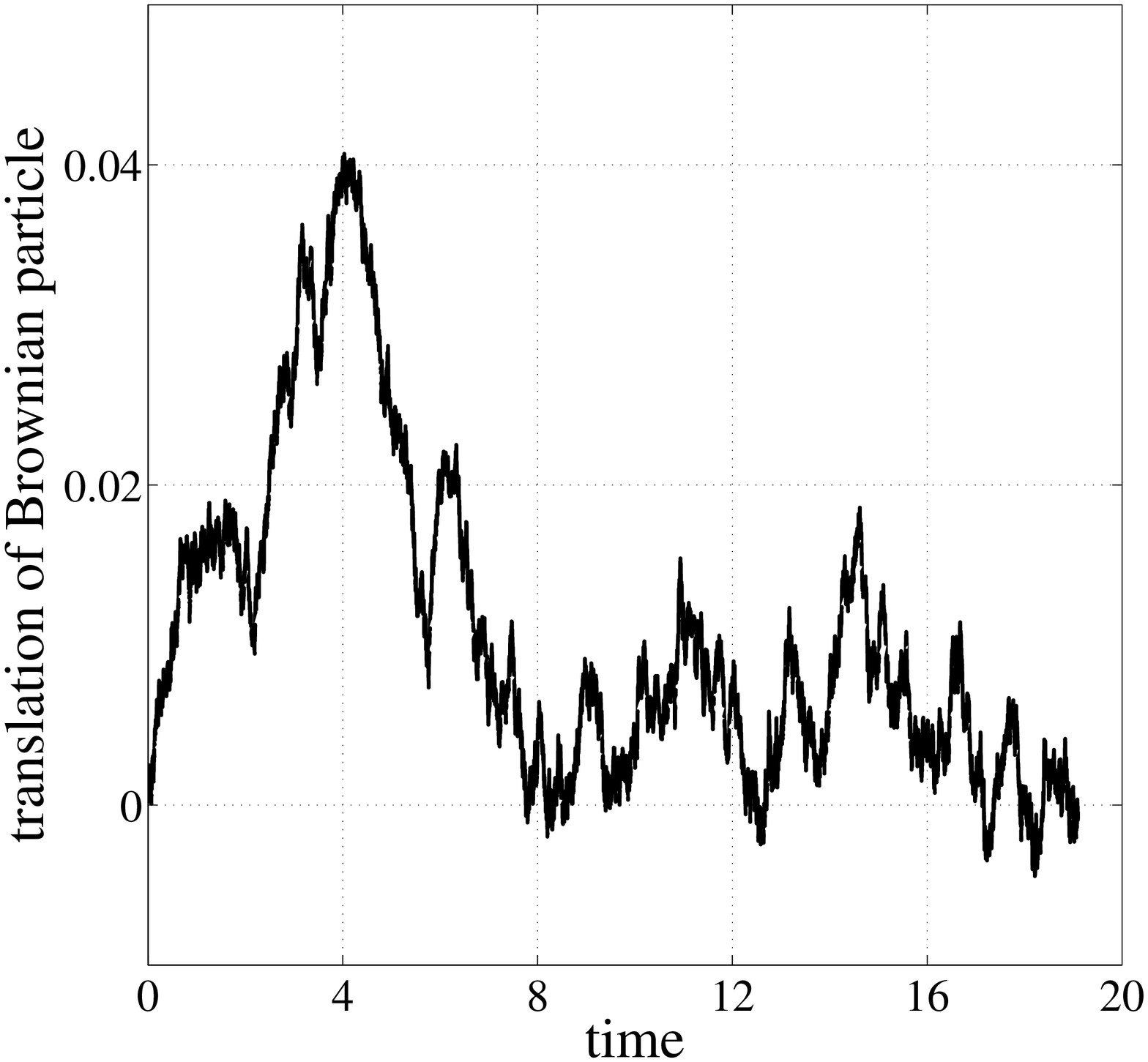}
\caption{\small  Left: position of the Brownian particle (with  zero load), as a function of time. Parameters:
the mass $m_0$ of the thermostat wall system is $10$; the Brownian particle mass is $m_1=100$ and
the gas molecule mass is  $m_2=1$. The  length of the circular rail track is $l=10^{-4}$ (the vertical axis measures the
positive or negative translation along the track) and the number of  events (an {\em event} being  defined as a collision 
between the two particles, a collision between the gas
molecule and  one of the walls, or simply the passage of the  two particles through a common position 
along the tracks without collision due to the obstacle's pin being down) is $N=10^6$.  The temperature parameters
are, from the middle graph to the top: $\sigma_1^2=1$ and $\sigma_2^2=1, 2, 4, 8$. For the lower graphs
the two parameters are reversed. A steady translation away from the hot wall and toward the cold wall is
apparent. 
On the right: another sample path obtained under the same conditions as the middle graph on the left.
 In particular, the two walls have the same temperature and there is no apparent rotation drift.
}
\label{brownianpaths}
\end{center}
\end{figure} 

These graphs suggest that the mass $m_1$ undergoes a noisy rotation, with speed of rotation that
depends on the difference in temperature between the walls.  When
the temperatures are switched, the direction of rotation is reversed.

When the two temperatures are equal, the Brownian particle appears to move  according
to mathematical Brownian motion.   See the right hand side of Figure \ref{brownianpaths}, which  shows
another sample path obtained under the same conditions as  the middle graph on the left hand side of the same figure.
 Viewed at this scale, the Brownian character of the motion is 
 more apparent.

The  effect of adding a constant force is shown in Figure \ref{meanvelocitiesload},
giving the mean velocity of rotation for a constant load
 while the temperature of one of the walls is changed. 
For relatively small temperature differences,
the Brownian particle rotates with constant mean velocity in the same direction of the force,
so
the work is done on the system.
When the temperature difference is sufficiently large the engine rotates against the force, so that work is done by the system.

The  efficiency of a heat  engine is traditionally defined, since Sadi Carnot's pioneering work,
as the (negative of the) ratio of the amount of mechanical work done by the system over the
heat taken from the heat source.  
The analysis of efficiency is based on a simple energy accounting.
At any given time $t>0$, let $Q_h(t)$ be the total amount of heat transferred to the system (gas-molecule plus Brownian particle)
since time $t=0$ due to collisions between the gas molecule and the hot wall. Let  $Q_c(t)$ be the heat similarly transferred 
to the system from the cold wall. These heats are obtained  by adding up the changes in kinetic energy of the gas molecule before
and after each collision.

\vspace{0.1in}
\begin{figure}[htbp]
\begin{center}
\includegraphics[width=3in]{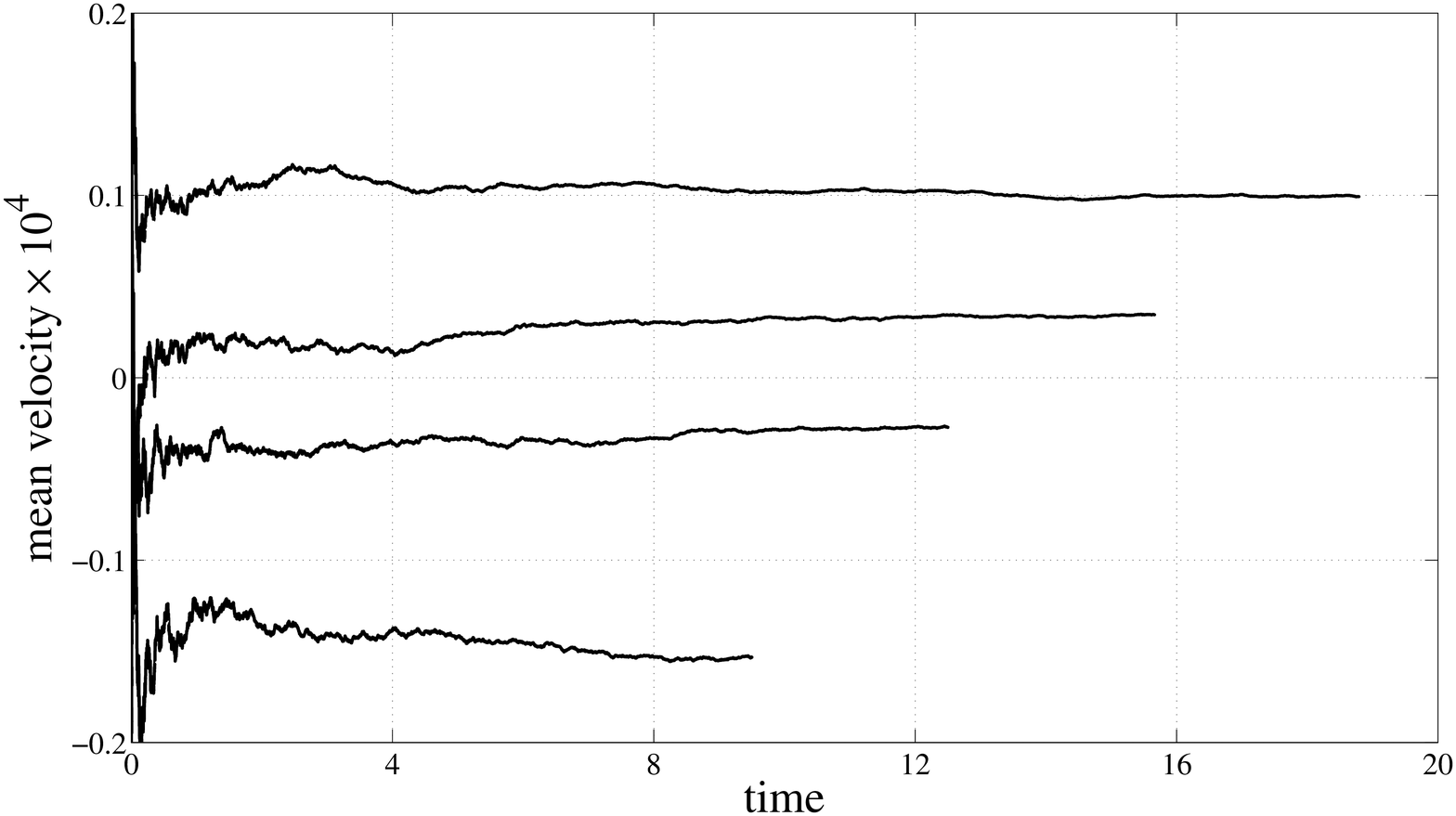}\ \ 
\caption{\small The working engine with a force load. The graphs show the mean velocity of 
the Brownian particle as a function of time. Common parameters for all graphs: 
The values of the masses $m_0, m_1, m_2$ are as in Figure \ref{brownianpaths}; 
the force load is $F=1$; the number of events (as explained in Figure \ref{brownianpaths}) is $N=10^6$;
the length of the track is $l=10^{-4}$; the temperature parameters, from top to bottom, are:
$\sigma_1^2=1$ and $\sigma_2^2=1, 2, 4, 8$.  }
\label{meanvelocitiesload}
\end{center}
\end{figure}

The (internal) energy of the system at time $t$ is $E_g(t)+E_b(t)$, where
$E_g(t)$ is the kinetic energy of the gas molecule and $E_b(t)$ is the kinetic energy of the Brownian particle.
The work done by a force $F$ as in Figure \ref{motor1} on the system up to time $t$    is denoted by $W(t)$. When $W(t)$ is negative,
we say that  work is done {\em by the system}. Recall that 
the   $W(t)=(x_b(t)-x_b(0))F$ for a constant $F$,  where $x_b(t)$ is the position of
the Brownian particle at time $t$. Over a time interval without collisions,   the change in $W$  equals  the change in kinetic energy
of the Brownian particle. Then the following identity holds:
\begin{equation}\label{balance} Q_h(t)+Q_c(t)+W(t)=E_g(t)-E_g(0)+E_b(t)-E_b(0).\end{equation} 
Now formally define the   (mean, at time $t$) {\em efficiency}  over one  sample history  of the engine, when work $W$ is negative  hence done {\em by} the system, as
\begin{equation}\label{effW}\epsilon_t(T_h,T_c)= -\frac{W(t)}{Q_h(t)},\end{equation}
which measures the fraction of heat transferred to the system from the hot wall that is converted to mechanical work over
the course of one history of the engine and is, therefore, a random variable.

Experimentally, we observe by running our Brownian engine that the quotient $(E_g(t)-E_g(0)+E_b(t)-E_b(0))/Q_h(t)$
goes to zero relatively quickly when the two temperatures are different. This is illustrated
in Figure \ref{efficiencydefinition}.

\vspace{0.1in}
\begin{figure}[htbp]
\begin{center}
\hspace{-0.1in}\includegraphics[width=2.5in]{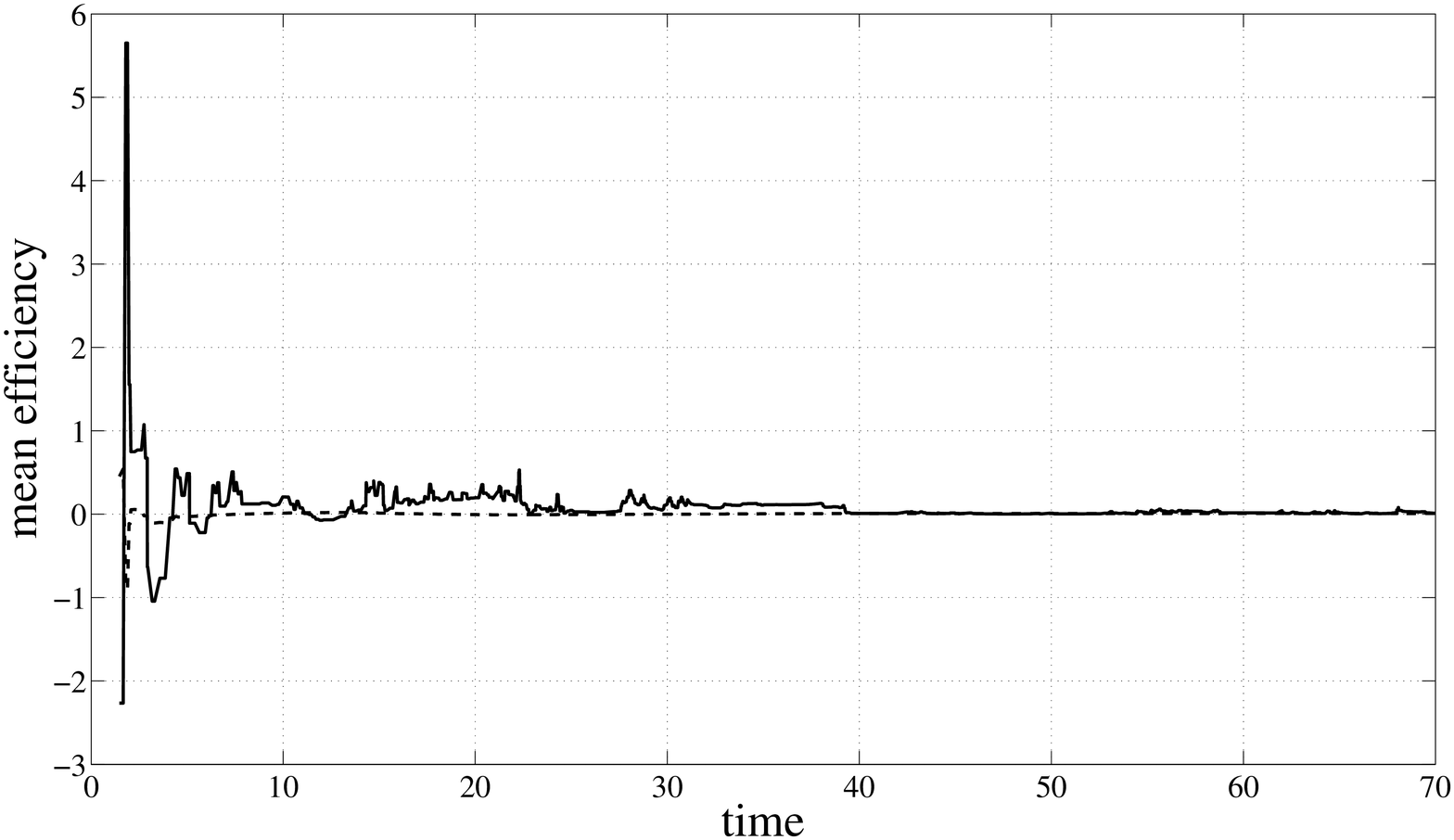}\hspace{-0.2in} \includegraphics[width=2.5in]{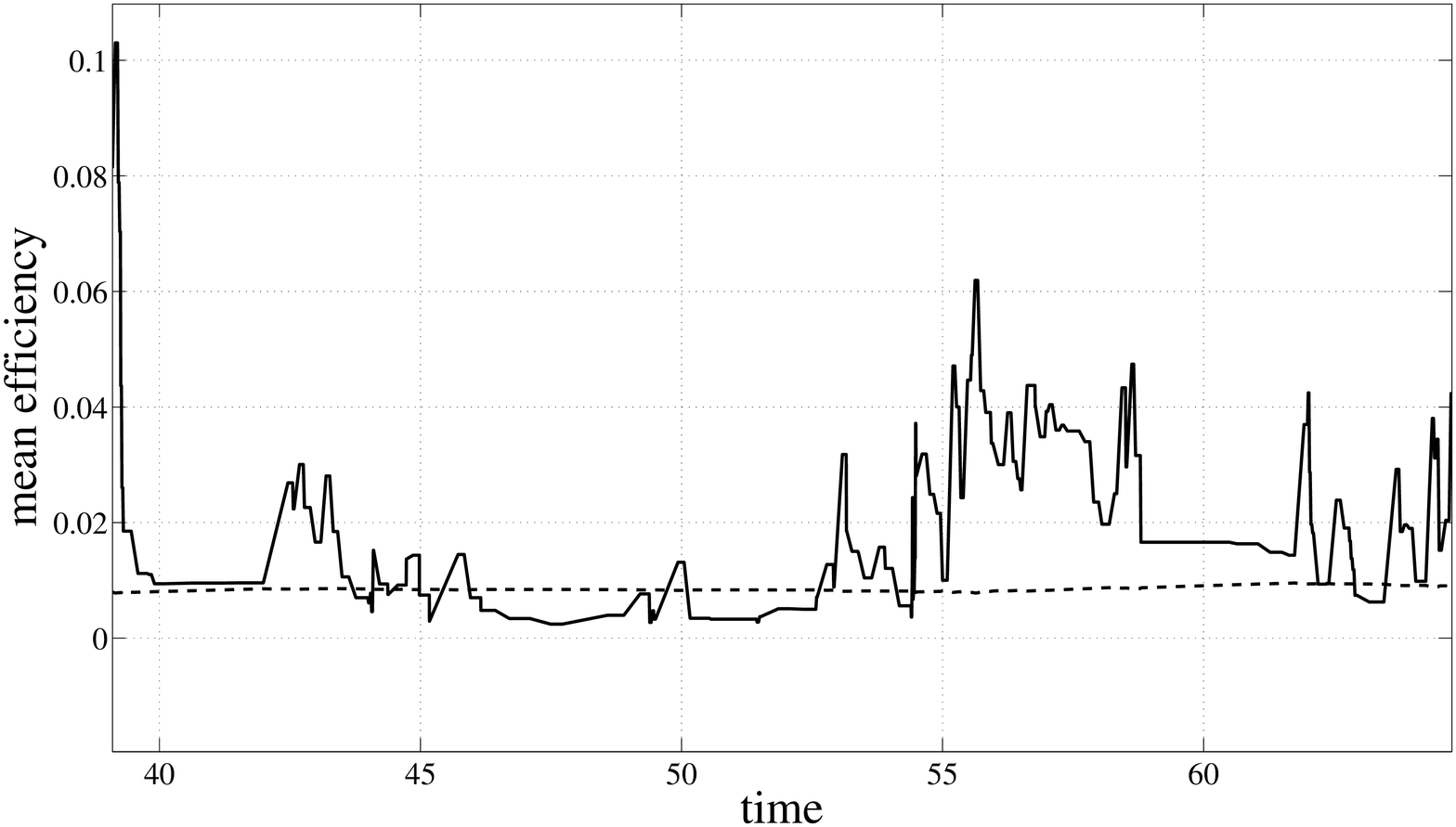}
\caption{\small Comparison of the definition of mean efficiency \ref{effW} (dashed line) and the alternative 
form \ref{effQ} (solid line).  We have applied a force $F=1$, the same masses as in Figure \ref{brownianpaths},
and temperature parameters $\sigma_1^2=1$, $\sigma_2^2=8$.  The graph shows a short run of $1000$ events. 
On the right, zooming in on part of the graph on the left shows that the efficiency is small but not zero.}
\label{efficiencydefinition}
\end{center}
\end{figure}

The efficiency measured at a steady operation regime may be expected
to equal (almost surely for large $t$) the alternative expression
\begin{equation}\label{effQ} \overline{\epsilon}_t(T_h,T_c)=1+\frac{Q_c(t)}{Q_h(t)}\end{equation}
where the two heats have opposite signs.

\vspace{0.1in}
\begin{figure}[htbp]
\begin{center}
\includegraphics[width=3.8in]{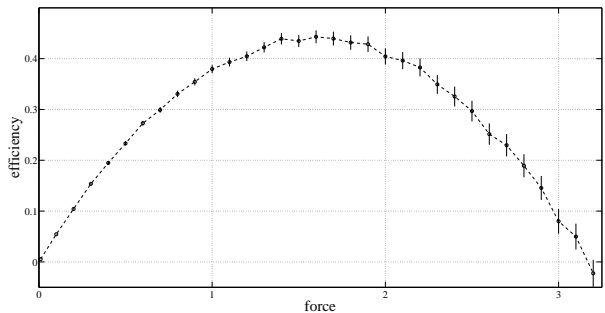}\ \ 
\caption{\small   The efficiency axis is in percentage units, so maximum efficiency
is a little below $0.4\%$. The vertical bars indicate $99\%$ confidence intervals. The parameters here are:
$m_0=10, m_1=100, m_2=1$,  $\sigma_1^2=1, \sigma_2^2=8$.   For each value of the
force we have evaluated the efficiency over $40000$ runs of the engine, each run of length  $2000$ elementary events.
(The dashed line connecting the mean values is there as a visual aid and has no significance.)}
\label{efficiencygraph}
\end{center}
\end{figure}

Compared to the classical upper limit of efficiency $1-T_c/T_h$  derived from
the second law of thermodynamics (for non-stochastic systems), our engine has very low efficiency. 
(See Figure \ref{efficiencygraph}.)
The engine can operate in the reverse direction: 
for a range of values of the force $F$ and the temperatures $T_c$ and $T_h$, 
work is positive (done to the system), with the effect of transferring heat from the
cold to the hot wall. In this regime, the engine operates as a {\em heat pump}.

We offer these informal numerical observations simply as evidence that  the engine functions  as
expected. A model of how a  detailed analysis of its operation  may be done  centered on the idea of entropy production  is 
the stochastic thermodynamic framework of \cite{seifert1}. The stochastic dynamic
of our engine is given by a Markov chain, so the first step in the analysis should
be to describe the process in terms of  Langevin equations by an appropriate scaling limit,
or pursue more directly the type of analysis of  \cite{qian}. These are tasks to be carried out in  a future paper.






\begin{thebibliography}{Abcdef}
  \bibitem{bali} P. B\'alint, N. Chernov, D. Sz\'asz, I. P. T\'oth, {\em Geometry of multidimensional dispersing billiards}. Ast\'erisque, 286 (2003), 119-150.
  \bibitem{buni} L.A. Bunimovich et al.,  {\em Dynamical Systems, Ergodic Theory and Applications}. Encyclopaedia of Mathematical Sciences, Vol. 100, Springer, 2000. 
  \bibitem{chernov} N. Chernov and R. Markarian, {\em Chaotic Billiards}. Mathematical Surveys and Monographs, Volume 127.
  Americal Mathematical Society, 2006.
  \bibitem{scott} S. Cook and R. Feres, {\em Random billiards with wall temperature and associated Markov chains}.
to appear in Nonlinearity,  arXiv:1202.2387, 2012.
 \bibitem{qian} Da-Quan Jiang, Min Qian, Min-Ping Qian, {\em Mathematical Theory of Nonequilibrium Steady States.}
  Lecture Notes in Mathematics, 1833, Springer, 2004.
  \bibitem{fy} R. Feres and G. Yablonsky, {\em Knudsen's cosine law and random billiards}. Chemical
  Engineering Science 59 (2004) 1541-1556.
  \bibitem{fz} R. Feres and H. Zhang, {\em The spectrum of the billiard Laplacian of a family of random billiards}, 
Journal of Statistical Physics, V. 141, N. 6 (2010) 1030-1054.
\bibitem{fz2} R. Feres and H. Zhang, {\em Spectral gap for a class of random billiards}. Commun.  Math. Phys. 313, 479-515 (2012).
\bibitem{gromov} M. Gromov, {\em Metric structures for Riemannian and non-Riemannian spaces}. Birkh\"auser, 2001.
\bibitem{gutkin} E. Gutkin, {\em Billiard dynamics: a survey with the emphasis on open problems}. Regular and Chaotic
Dynamics {\bf 8} (1), 2003.
\bibitem{katok} A. Katok and B. Hasselblatt, {\em Introduction to the modern theory of dynamical systems}.
Encyclopedia of Mathematics and its Applications 54, Cambridge University Press, 1995.
\bibitem{knudsen} M. Knudsen, {\em Kinetic theory of gases\----some modern aspects}. Methuen's Monographs on
Physical Subjects, London, 1952.
\bibitem{seifert1} U. Seifert, {\em Stochastic thermodynamics}. In Lecture Notes:`Soft Matter, From Synthetic to Biological Materials.'
39th IFF Spring School, Institute of Solid State Research, Research Centre J\"ulich, 2008.
\bibitem{seifert2} U. Seifert, {\em Stochastic thermodynamics: principles and perspectives}.  arXiv:0710.1187, 2007.
\bibitem{steckline} V.S. Steckline, {\em Zermelo Boltzmann, and the recurrence paradox}.  Am. J. Phys. {\bf 51} (10), October 1983.
\bibitem{zheng} J. Zheng, Z. Zheng, C. Yam, G. Chen, {\em Computer simulation of Feynman's ratchet and pawl system}.
Physical Review E{\bf 81}, 061104, 2010.
\end{thebibliography}







\end{document}